\pdfoutput=1
\documentclass[a4paper,UKenglish,cleveref, autoref,thm-restate]{lipics-v2021}
\hideLIPIcs
\usepackage{booktabs}
\usepackage{textgreek}
\usepackage{graphicx}
\usepackage{framed}
\usepackage{caption}
\usepackage{tikz-cd}
\usepackage[unicode]{hyperref}
\usepackage{doi}
\Crefname{remark}{Remark}{Remarks}
\Crefname{observation}{Observation}{Observations}
\usepackage{thm-restate}
\usepackage{algorithm}
\usepackage{algorithmicx}
\usepackage{algpseudocode}
\usepackage{todonotes}
\usepackage{xspace}

\usepackage{microtype}
\usepackage{xcolor}
\usepackage{amsmath}
\usepackage{amssymb}
\usepackage{amsthm}
\usepackage{enumitem}
\usepackage{mathtools}

\newcommand{\goodInput}{non-prime-divisible and order-invariant}

\DeclareMathOperator{\poly}{poly}

\newcommand{\LOCAL}{\ensuremath{\textsf{LOCAL}}\xspace}
\newcommand{\CONGEST}{\ensuremath{\textsf{CONGEST}}\xspace}
\newcommand{\WFLOC}{\ensuremath{\textsf{ASYNC LOCAL}}\xspace}

\definecolor{darkgreen}{rgb}{0,0.5,0}
\definecolor{darkred}{rgb}{0.4,0,0}
\hypersetup{
	colorlinks=true,
	linkcolor=darkred,
	citecolor=darkgreen,
	filecolor=black,
	urlcolor=[rgb]{0,0.1,0.5},
	pdftitle={Asynchronous Fault-Tolerant Distributed Proper Coloring of Graphs},
	pdfauthor={TBD}
}

\crefname{algocf}{Alg.}{Algs.}
\Crefname{algocf}{Algorithm}{Algorithms}

\graphicspath{{figs/}}

\newcommand{\stateold}[0]{\ensuremath{\mathrm{OldState}}}
\newcommand{\statenew}[0]{\ensuremath{\mathrm{NewState}}}
\newcommand{\stateinit}[0]{\mathrm{Init}}
\newcommand{\statenext}[0]{\mathrm{Alg}}
\newcommand{\stateterm}[0]{\mathrm{Terminated}}
\newcommand{\imm}[0]{\mathrm{WriteSnapshot}}
\newcommand{\id}[1]{\mathrm{id}_{#1}}
\newcommand{\inp}[1]{\mathrm{input}_{#1}}

\addtocounter{footnote}{0}

\newcommand*{\sign}[1]{\operatorname{sign}(#1)}

\title{Asynchronous Fault-Tolerant Distributed Proper Coloring of Graphs} 

\titlerunning{Asynchronous Fault-Tolerant Distributed Proper Coloring of Graphs}

\author{Alkida Balliu}{Gran Sasso Science Institute, L'Aquila, Italy}{alkida.balliu@gssi.it}{}{}
\author{Pierre Fraigniaud}{IRIF - CNRS \& Univ. Paris Cité, France}{pierre.fraigniaud@irif.fr}{}{}
\author{Patrick Lambein-Monette}{Unaffiliated}{patrick@lambein.name}{}{}
\author{Dennis Olivetti}{Gran Sasso Science Institute, L'Aquila, Italy}{dennis.olivetti@gssi.it}{}{}
\author{Mikaël Rabie}{IRIF - Université Paris Cité, France}{mikael.rabie@irif.fr}{}{}

\authorrunning{A. Balliu, P. Fraigniaud, P. Lambein-Monette, D. Olivetti, M. Rabie} 

\Copyright{Alkida Balliu, Pierre Fraigniaud, Patrick Lambein-Monette, Dennis Olivetti, Mikaël Rabie} 
\ccsdesc[500]{Theory of computation~Distributed algorithms}

\keywords{\LOCAL model, Graph Coloring, Renaming, Weak Symmetry-Breaking, Fault-Tolerance, Wait-Free Computing}

\funding{Partially funded by MUR (Italy) Department of Excellence 2023 - 2027, the PNRR MIUR research project GAMING ``Graph Algorithms and MinINg for Green agents'' (PE0000013, CUP D13C24000430001), by the research project RASTA ``Realtà Aumentata e Story-Telling Automatizzato per la valorizzazione di Beni Culturali ed Itinerari'' (Italian MUR PON Project ARS01 00540), and by the French ANR projects DUCAT (ANR-20-CE48-0006) and QuDATA (ANR-18-CE47-0010).
}

\begin{document}

\maketitle

\begin{abstract}
We revisit \emph{asynchronous} computing in networks of \emph{crash-prone} processes, under the asynchronous variant of the standard \LOCAL model, recently introduced by Fraigniaud et al. [DISC 2022]. We focus on the vertex coloring problem, and our contributions concern both lower and upper bounds for this problem.

On the upper bound side, we design an algorithm tolerating an arbitrarily large number of crash failures that computes an $O(\Delta^2)$-coloring of any $n$-node graph of maximum degree~$\Delta$, in $O(\log^\star n)$ rounds. This extends Linial's seminal result from the (synchronous failure-free) \LOCAL model to its asynchronous crash-prone variant. 
Then, by allowing a dependency on $\Delta$ on the runtime, we show that we can reduce the colors to $\big(\frac12(\Delta+1)(\Delta+2)-1 \big)$.
For cycles (i.e., for $\Delta=2$), our algorithm achieves a 5-coloring of any $n$-node cycle, in $O(\log^\star n)$ rounds. This improves the known 6-coloring algorithm by Fraigniaud et al., and fixes a bug in their algorithm, which was erroneously claimed to produce a 5-coloring. 

On the lower bound side, we show that, for $k<5$, and for every prime integer~$n$, no algorithm can $k$-color the $n$-node cycle in the asynchronous crash-prone variant of \LOCAL, independently from the round-complexities of the algorithms. This lower bound is obtained by reduction from an original extension of the impossibility of solving \emph{weak symmetry-breaking} in the wait-free shared-memory model. We show that this impossibility still holds even if the processes are provided with inputs susceptible to help breaking symmetry.

\end{abstract}

\ccsdesc[500]{Theory of computation~Distributed algorithms}

\keywords{\LOCAL model, Graph Coloring, Renaming, Weak Symmetry-Breaking, Fault-Tolerance, Wait-Free Computing}

\maketitle

\section{Introduction}

\subsection{Asynchrony, Failures, and Networks}

To what extent a global solution to a computational problem can be obtained from locally available data? What can be computed locally? These are some of the questions that were asked, and partially answered 30 years ago in two seminal papers \cite{Linial92,NaorS95} in the field of distributed network computing. Since then, tremendous progress has been made about these questions, and even detailed books~\cite{HirvonenS2020,Peleg2000} can only touch a small fraction of the content of the current literature on this topic. Nevertheless, the vast majority of the achievements on \emph{local computing} have been obtained in \emph{synchronous failure-free} models, among which the most common ones are referred to as \LOCAL~\cite{Linial92} and \CONGEST~\cite{Peleg2000}. 

In both models, processing nodes occupy the vertices of a graph, and exchange messages along the edges of that graph. They all start at the same time, and computing proceeds as a sequence of synchronous rounds. At each round, every pair of adjacent nodes can exchange messages (one in each direction), and every node can perform some individual computation. \CONGEST differs from \LOCAL only as far as the message size is concerned: messages are bounded to  be of size at most~$B$ bits in \CONGEST (it is common to set $B=O(\log n)$). There are at least two solid reasons why such elegant but simplistic models should be considered. First, they ideally capture the notion of \emph{spatial locality}, as algorithms performing in $t$ rounds produce an output at each node that is solely based on the $t$-neighborhood of the node. Second, the existence of efficient synchronizers~\cite{AwerbuchPPS92,AwerbuchP90b,GhaffariT23} enables to implement algorithms designed for synchronous models on asynchronous networks, with only limited slowdown. 

Yet, models such as \LOCAL and \CONGEST suffer from one notable limitation: they ignore the potential presence of failures. Indeed, transient failures have been addressed in the framework of \emph{self-stabilization}, but crash or malign  failures are mostly ignored in the framework of local computing in networks. Instead, studying the interplay of asynchrony and failures has been the main topic of interest of distributed computing in general~\cite{AttiyaW04,Lynch96,Raynal18}, since the seminal ``FLP impossibility result'' stating that consensus is impossible in asynchronous systems with failures, even under the restriction that at most one crash failure may occur~\cite{FischerLP85}. However, the design of algorithms dedicated to asynchronous crash-prone systems have been mostly performed in \emph{shared-memory} or \emph{message-passing} models: the former assumes that processes exchange information by writing and reading in a shared memory; the latter assumes that any two processes can exchange messages directly along a private channel. While these two models are excellent abstractions of very many types of distributed systems, ranging from multi-core architectures to large-scale computing platforms, they do not enable the study of spatial locality, as the structure of the physical network is abstracted away. 

An attempt to resolve this tension between synchronous failure-free computing in networks, and asynchronous computing in crash-prone systems has been recently proposed~\cite{FraigniaudLR22}, by considering asynchronous networks subject to crash failures. 

\subsection{The ASYNC LOCAL Model}

The asynchronous crash-prone \LOCAL model\footnote{One could also consider the variant \textsf{ASYNC CONGEST} of \WFLOC by limiting to $O(\log n)$ bits the size of the registers in which nodes read and write, but we restrict ourselves to the \LOCAL variant, as standard wait-free computing does not generally restrict the size of the registers. } (\WFLOC in short),
  introduced in \cite{FraigniaudLR22},
  aims at capturing a setting
  that is a hybrid between shared memory and network computing.
This model can be described conceptually in two possible ways
  (see Section~\ref{ssec:model} for more details): 
\begin{itemize}
    \item The \WFLOC model can be viewed as the standard wait-free shared-memory model~\cite{AttiyaW04,Herlihy2013} in which the read-access to other process's registers is restricted. It bears similarities with the \emph{atomic state} model in self-stabilization~\cite{BlinT13}. In an $n$-process system, each process $i\in[n]$ can solely read the registers of processes $j\in N_G(i)$, where $N_G(i)$ denotes the set of neighbors of vertex~$i$ in a graph~$G$. That is, the wait-free shared-memory model is the \WFLOC model in which the graph $G$ is fixed to be the complete graph (or clique)~$K_n$. 
    
    \item The \WFLOC model can alternatively be viewed as the standard \LOCAL model~\cite{HirvonenS2020,Peleg2000} in which each node writes in its local register(s) instead of sending messages, and reads the registers of its neighbors instead of receiving messages from them. In addition, \WFLOC allows asynchronous executions, that is, each process reads and writes at its own pace, which may vary with time, and it may even crash (i.e., stop functioning, and never recover). Note that, as for \LOCAL, the graph~$G$ is unknown to the nodes in \WFLOC, as it is typically the \emph{input} to the problems of interest in network computing.
\end{itemize}

In the framework of asynchronous computing, the computing elements are referred to as \emph{processes}, whereas they are referred to as \emph{nodes} in the context of computing in networks, but we use these two terms indistinctly. The terminology ``wait-free'' refers to the fact that (1)~an arbitrarily large number of processes can crash, and (2)~a node cannot distinguish whether a neighboring node has crashed or is simply slow, from which it follows that a node must never ``wait'' for some action performed by another node, and must terminate independently from which of the other nodes have crashed (unless itself has crashed). 

\medskip

It was shown in~\cite{FraigniaudLR22} that the computing power of \WFLOC is radically different from the one of \LOCAL. Indeed, the authors proved that constructing a maximal independent set (MIS) is simply \emph{impossible} in \WFLOC, even in the $n$-node cycles $C_n$, $n\geq 3$, while, on cycles, it just takes $\Theta(\log^\star n)$ rounds in \LOCAL~\cite{ColeV86,Linial92}. However, the authors show also that proper coloring $C_n$ is possible in \WFLOC, to the expense of using a larger palette of colors, i.e., 6~colors instead of just 3 as in \LOCAL (a 5-coloring algorithm is also claimed in \cite{FraigniaudLR22}, but, as we shall show later, there is a bug in that algorithm). Indeed, a simple reduction to \emph{renaming} (see~\cite{AttiyaW04} for the definition) shows that, under the \WFLOC model, no algorithms can proper color all graphs of maximum degree~$\Delta$ using less than $2\Delta+1$ colors whenever $\Delta+1$ is a power of a prime. This is because \WFLOC and standard shared-memory coincides when the graph is a clique of $n=\Delta+1$ nodes. The main result in~\cite{FraigniaudLR22} is a distributed asynchronous algorithm in the \WFLOC model that  achieves proper 6-coloring of any $n$-node cycle, $n\geq 3$, in $O(\log^\star n)$ rounds, which is optimal thanks to~\cite{Linial92}. In \WFLOC, the \emph{round-complexity} of an algorithm is the maximum, taken over all nodes, and all executions, of the number of times a node writes in its register, and reads the registers of its neighbors.  

\subsection{Our results}

In a nutshell, we show that there exists an algorithm for proper coloring graphs with maximum degree $\Delta$ in the \WFLOC\/ model, using a palette of $\frac12(\Delta+1)(\Delta+2)-1$ colors, resulting into a 5-coloring algorithm for the cycles. This result was obtained by first showing how to implement Linial's coloring algorithm in the asynchronous setting, and then by developing a new technique based on reallocating identifiers to nodes. Note that even implementing Linial's coloring algorithm asynchronously is not straightforward, as it is not even clear whether the trivial recoloring algorithm that proceeds iteratively over all color classes can be implemented in the \WFLOC\/ model. Moreover, we show that, for infinitely many values of~$n$, 5-coloring the $n$-node cycles is the best that can be achieved in \WFLOC. This significantly improves the lower bound in~\cite{FraigniaudLR22} on the number of colors required for proper coloring cycles under \WFLOC, which held for $n=3$ only. 

Obtaining our lower bound required to revisit entirely the known lower bound on weak symmetry breaking\footnote{Weak symmetry breaking is the task in which processes start with no inputs, and each process must output~0 or~1, under the constraint that, whenever all processes terminate, at least one process must output~0, and at least one process must output~1. } in the standard asynchronous shared-memory model, by considering the impact of a priori ``knowledge'' given to the processes. For instance, if the processes know a priori that one process is given advice~0, and one process is given advice~1, then weak symmetry breaking becomes trivially solvable. For which a priori knowledge weak symmetry breaking becomes trivially solvable, and for which it remains unsolvable? We show that answering this novel question for specific types of a priori knowledge results into new impossibility results for the standard asynchronous shared-memory model, which translate into lower bounds and impossibility results in the \WFLOC\/ model. 

We stress the fact that while all (Turing computable) tasks are solvable in the \LOCAL\/ model, not all taks are solvable in \WFLOC, yet we also address \emph{complexity} issues, by showing that, for constant~$\Delta$, our $\big(\frac12(\Delta+1)(\Delta+2)-1\big)$-coloring algorithm performs in $O(\log^\star n)$ rounds in \WFLOC, that is, as fast as the $\Omega(\log^\star n)$ lower bound~\cite{Linial92} on the number of rounds required for coloring cycles in the synchronous failure-free \LOCAL\/ model. These results are detailed next. 

\subsubsection{Proper Coloring}

We mostly focus on distributed proper coloring, arguably one of the most important and thoroughly studied symmetry-breaking tasks in network computing --- see, e.g., \cite{FuchsK23,GhaffariK21,HalldorssonKNT22} for recent results on the matter\footnote{In the context of distributed computing in networks, especially in the \LOCAL\/ and \CONGEST\/ models, one is interested in properly coloring graphs with maximum degree~$\Delta$ using a palette of $f(\Delta)$ colors, where $f(\Delta)$ grows slowly with~$\Delta$. One typical example is $f(\Delta)=\Delta+1$ as all graphs of maximum degree $\Delta$ can be properly colored with $\Delta+1$ colors, but one is also interested in larger functions $f$, e.g., $f(\Delta)=\Theta(\Delta^2)$, whenever this choice enables to obtain faster algorithms.}. First, we show that Linial's technique from~\cite{Linial92} based on cover-free families of set systems can be used asynchronously, for the design of an $O(\Delta^2)$-coloring of graphs of maximum degree~$\Delta$, running in $O(\log^\star n)$ rounds in $n$-node graphs under \WFLOC. Then we show that the approach from~\cite{FraigniaudLR22} for 6-coloring cycles can be generalized to color arbitrary graphs. Specifically, we design an algorithm computing a $\frac{(\Delta+1)(\Delta+2)}{2}$-coloring in graphs of maximum degree~$\Delta$ running in $O(\log^\star n)+f(\Delta)$ rounds under \WFLOC, where the additional term $f(\Delta)$ depends on $\Delta$ only. This line of results culminates in the design of an algorithm enabling to save one color, i.e., that computes a $\big(\frac{(\Delta+1)(\Delta+2)}{2}-1\big)$-coloring, still running in $O(\log^\star n)+f(\Delta)$ for some function~$f$. Reducing the color palette by just one color may seem of little importance, but it is not, for two reasons. First, a palette of size $\frac{(\Delta+1)(\Delta+2)}{2}-1$ is the best that we are aware of for which it is possible to proper color all graphs of maximum degree~$\Delta$ in $O(\log^\star n)$ rounds in \WFLOC (ignoring the additional term depending on~$\Delta$ only). Saving one more color appears to be challenging. Second, in the case of cycles, i.e., $\Delta=2$, this allows us to fix a bug in the 5-coloring algorithm from~\cite{FraigniaudLR22}. Indeed, this latter algorithm is shown to be erroneous, as there are schedulings of the nodes that result in livelocks preventing the algorithm from terminating. Nevertheless, our algorithm shows that 5-coloring the $n$-node cycles in $O(\log^\star n)$ rounds under \WFLOC is indeed possible. 

\subsubsection{Lower Bounds and Impossibility Results}

Our second line of contribution is related to lower bounds on the size of the color palette enabling to proper color graphs asynchronously. It was observed in~\cite{FraigniaudLR22} that since the class of graphs with maximum degree~$\Delta$ includes the clique with $n=\Delta+1$ nodes, and since \emph{renaming}~\cite{AttiyaW04} in a set of less than $2N-1$ names cannot be done wait-free in $N$-process shared-memory systems whenever $N$ is a power of a prime, proper coloring graphs of maximum degree~$\Delta$ in \WFLOC cannot be achieved with a color palette smaller than $2\Delta+1$ colors, i.e., 5~colors in the case of cycles (independently from the number of rounds). However, the question of whether one can 4- or even 3-color long cycles (i.e., excluding the specific case of the clique~$C_3$) under \WFLOC was  left open in~\cite{FraigniaudLR22}. We show that this is impossible whenever $n$ is prime, that is, there are infinitely many values of $n$ for which 5-coloring the $n$-node cycle is the best that can be achieved in \WFLOC. 

\subsubsection{Reduction from Weak Symmetry-Breaking with Inputs}

We achieve our lower bound on the number of colors thanks to a result of independent interest in the standard framework of wait-free shared-memory computing. We show that there are no symmetric wait-free algorithms solving \emph{weak symmetry-breaking}~\cite{AttiyaW04} in $n$-process asynchronous shared-memory systems whenever $n$ is prime, \emph{even if processes are provided with inputs from a \goodInput\/ set of inputs}. We achieve this impossibility result by extending the proof in~\cite{AttiyaP16} for weak symmetry-breaking to the case in which processes have inputs that do not trivially break symmetry. Our impossibility result for weak symmetry-breaking with inputs has other consequences on the \WFLOC model, including the facts that weak 2-coloring is impossible in cycles of prime size, and that, for every even $\Delta\geq 2$, there is an infinite family of regular graphs for which $(\Delta+2)$-coloring cannot be solved in \WFLOC. 

Finally, using different techniques, we also show that even a weak variant of maximal independent set (MIS) cannot be solved in cycles with at least 7~nodes, and that, for every $\Delta\geq 2$,  $(\Delta+1)$-coloring trees of maximum degree~$\Delta$ is impossible under \WFLOC. 

\subsection{Related Work}

The combination of asynchrony \emph{and} failures in the general framework of distributed computing in networks has been studied a lot in the context of \emph{self-stabilization}. The latter deals with \emph{transient} failures susceptible to modify the content of some of the variables defining the states of the nodes. The role of a self-stabilizing algorithm is therefore to guarantee that if the network is in an illegal configuration (i.e., a configuration not satisfying some specific correctness condition), then it will automatically return to a legal configuration, and will remain in a legal configuration, unless some other failure(s) occur. Self-stabilizing graph coloring algorithms have been designed~\cite{BarenboimEG18,BernardDPT09,BlairM12,BlinFB19}. However, these algorithms provide solutions only for executions during which there are no failures. Instead, in \WFLOC, failures may occur at any time during the execution, and once a process crashes it never recovers. This has important consequences on what can or cannot be computed in \WFLOC. For instance, 3-coloring the $n$-node cycle is possible in a self-stabilizing manner for every $n\geq 3$, while we show that even 4-coloring  the $n$-node cycle is impossible for infinitely many~$n$ (namely, for all prime~$n$). 

It is also worth mentioning \cite{CastanedaDFRR19,Delporte-Gallet19}, which introduced the \textsf{DECOUPLED} model, where crash-prone processes occupy the nodes of a \emph{reliable} and \emph{synchronous} network. The \textsf{DECOUPLED} model is stronger than \WFLOC, and indeed it was shown that if there exists an algorithm solving a task in the \LOCAL model, then there exists an algorithm solving that task in the \textsf{DECOUPLED} model as well, with limited slowdown. Instead, we show that even a weak variant of MIS is impossible in large cycles under \WFLOC. 

Another field of research very much related to our work is the study of \emph{synchronous} networks with failures, whether it be crash or even malicious process failures, or message omission failures (see, e.g., \cite{CastanedaFPRRT23,Nowak0W19,SantoroW89,WinklerPG0023}). In these models, the focus has mostly been put on the study of tasks such as consensus and set-agreement. The  \WFLOC\/ model somehow mixes some of the key aspects of the models considered in these work, including the presence of crash failures, and the fact that the communications are mediated by a graph distinct from the complete graph. The same way standard wait-free computing in shared-memory systems can be viewed as one specific instance of the oblivious message adversary model, wait-free computing in the \WFLOC\/ model in a graph~$G$ may be viewed as the instance of the oblivious message adversary model in which messages can only be sent along the edges of the graph~$G$. We however focus on solving graphs problems such as coloring or independent set, motivated by the need to solve various symmetry breaking problems in networks, including frequency assignment and cluster decomposition. For such problems, it is more more convenient to use the framework of \WFLOC, in which the graph $G$ is part of the input, as in the \LOCAL\/ model. 

\section{Model and Definitions}

We first recall the \WFLOC model as introduced in~\cite{FraigniaudLR22}, and then provide an example for an algorithm in this model. 

\subsection{The ASYNC LOCAL model}\label{ssec:model}

Like the \LOCAL model~\cite{Peleg2000},
    the \WFLOC model assumes a set of $n \geq 1$ processes,
    each process occupying a distinct node
    of an $n$-vertex graph~$G = (V,E)$,
    which is supposed to be simple and connected.
Each process, i.e.,
    each node $v\in V$,
    has an identifier $\id{v}$
    that is supposed to be unique in the graph.
The identifiers are not necessarily between~1 and~$n$,
    but they are supposed to be stored on $O(\log n)$ bits.
That is,
    all node identifiers lie in the integer interval~$[1, N]$
    for some bound~$N = \operatorname{poly}(n)$.
Like in the asynchronous shared-memory model,
    every node~$v$ comes equipped
    with a single-writer/multiple-reader register~$R(v)$
    in which it can write values.
However, in contrast with the shared-memory model,
    \emph{only $v$'s neighbors in the graph~$G$}
    are able to read its register~$R(v)$,
    and symmetrically,
    node~$v$ can only read the registers~$R(w)$
    of nodes~${w \in N_G(v)= \{u \in V \mid \{u,v\} \in E\}}$.
We assume that each node can
    write in its register,
    and then read all its neighbors' registers,
    in a single atomic operation.
Neighboring nodes can perform this write\&{}read operation concurrently,
    in which case they both read the value concurrently written
    in the other node's register.
This communication primitive is thus akin
    to an \emph{immediate snapshot} object
    with read accesses mediated by a graph,
    in a similar manner to the \emph{atomic state} model
    in the context of self-stabilizing algorithms~\cite{BlinT13}.
Computation proceeds asynchronously,
    and each node may crash,
    in which case it stops functioning,
    and it never recovers.
Therefore,
    in the particular case of the clique~$G=K_n$,
    \WFLOC boils down to the standard
    asynchronous crash-prone shared-memory
    model with immediate snapshots~\cite{AttiyaW04}. 
The registers are of unbounded size.
Therefore,
    as in the \LOCAL model,
    and as in most wait-free computing models~\cite{Herlihy2013} as well,
    we can assume \emph{full-information protocols},
    in which every node writes its entire state in its register,
    and read the states of its neighbors in their registers.

\medbreak

\noindent\emph{Remark.}  Due to its nature, the \WFLOC\/ model may have also been named ``iterated immediate local snapshot''. Nevertheless, for its close connection to the standard \LOCAL\/ model used for the study of graph problems (e.g., coloring) in distributed computing, we preferred to stick to the terminology \WFLOC.    

\subparagraph{Input.}

In addition to its identifier $\id{v}$, every node~$v$ may be provided with some input, denoted by $\inp{v}$. The latter may be the number $n$ of nodes in the graph, or an upper bound~$N$ on~$n$, or any label $\ell(v)\in\{0,1\}^*$ whose semantic depends on the context (e.g., it may represent a boolean mark, or  a color, etc.). Note that the network $G$ is typically unknown to the nodes, even if some specific parameters may be provided to each node as input, such as the maximum degree $\Delta$ of~$G$. 

\subparagraph{Algorithm.}

An algorithm $\mathcal{A}$ for the \WFLOC model may be described by two functions:
\begin{itemize}
    \item $\stateinit$: used to initialize the state of each node, as a function of its input;
    \item $\statenext$: used to update the state of a node, as a function of its current state, and of the states of its neighbors.
\end{itemize}

\subparagraph{Scheduling.}

An execution of an algorithm  $\mathcal{A}$ depends on how the nodes are scheduled.
A \emph{scheduling} is a sequence $\mathcal{S} = S_1,S_2,\ldots$ of subsets $S_i \subseteq V$ of nodes. For every $i\geq 1$, the set $S_i$ denotes the set of nodes that are activated at \emph{step}~$i$. Each of these nodes performs an immediate-snapshot, and updates its state accordingly. For instance, the scheduling $\{u,v\},\{v\},\{v\},\{v\},\dots$  represents the execution in which nodes $u$ and $v$ run concurrently at the first step, and then $v$ runs solo, i.e., $v$ is the only node activated at every step $i\geq 2$. That is, $u$ has crashed after step~1, and all the nodes $w\notin\{u,v\}$ had crashed initially, none of them taking any step. Instead, the scheduling $V,V,V,\dots$ represents a synchronous execution in which no node crashes.   

\subparagraph{Full-Information Protocols.}
For every $v \in V$, let 
$
\stateold_{v,1} \gets \bot, \;\mbox{and}\; \statenew_{v,1} \gets \stateinit(\id{v},\inp{v}).
$
For every $i\geq 1$, the variable $\stateold_{v,i}$ represents what a neighbor of $v$ gets whenever reading the memory of $v$, and  $\statenew_{v,i}$ represents the updated state of $v$, which will become visible to its neighbors the next time $v$ is scheduled. 
More specifically, for every $i\geq 1$, if $v \notin S_i$, then 
$\stateold_{v,i+1} \gets \stateold_{v,i}$ and $\statenew_{v,i+1} \gets\statenew_{v,i}$.
Instead, if  $v \in S_i$, then $\stateold_{v,i+1} \gets \statenew_{v,i}$, and 
$
\statenew_{v,i+1} \gets \statenext(\stateold_{v,i+1},\{\stateold_{u,i+1} \mid u \in N_G(v)\}).
$
In other words, all nodes that are scheduled at step $i$ write their current state, then read the state of their neighbors, and then use the obtained knowledge in order to update their state. The new states resulting from these updates will become visible to their neighbors the next time that they are scheduled. That is, we model a setting in which writing and then reading the state of the neighbors is an atomic operation, but it may take some time to compute a new state.

\subparagraph{Termination.}

We let $\stateterm(x)$ be a special state denoting that a node terminates with output~$x$.
If a node $v$ satisfies $\statenew_{v,t} = \stateterm(x)$ at some step $t\geq 1$, then $v$ decides the output~$x$, and it is assumed that if $v$ is scheduled again in the future, then its state does not change, that is, $\statenew_{v,t+i} = \statenew_{v,t}$ for all $i \ge 1$.

\subparagraph{Round complexity.}

The runtime of a node $v$ is defined as 
\[
T_v = |\{i\geq 1 \mid v \in S_i \text{ and } \statenew_{v,i} \neq \stateterm(x) \text{ for any  possible output } x\}|.
\] 
That is, the runtime of a $v$ is equal to how many times $v$ is scheduled before it terminates. The runtime of an algorithm on a graph $G = (V,E)$ is then $\max\{T_v \mid v \in V\}$. The runtime of an algorithm in a graph class $\mathcal{G}$ is the maximum runtime of the algorithm, over all graphs $G\in\mathcal{G}$. The runtime of an algorithm may depend on the identifiers given to the nodes. However, as said before, we use the standard assumption that the identifiers are from the interval $[1,N]$ where $N=\mbox{poly}(n)$. The runtime is thus typically expressed as a function of $n$ (the order of the graph) and $\Delta$ (the maximum degree of the graph). The complexity of a problem is the minimum runtime (as a function of $n$ and $\Delta$) among all possible algorithms that solve the problem. The typical graph class we are interested in is~$\mathcal{G}_\Delta$, the class of all graphs with maximum degree~$\Delta$.

\medbreak

\noindent\emph{Remark.} In absence of failures, and if all nodes run synchronously, the runtime of an algorithm in the \WFLOC\/ model is identical to its runtime in the \LOCAL\/ model. 

\subsection{Algorithm Description}

While an algorithm can be formally described by providing the two functions $\stateinit$ and $\statenext$, we now describe an alternative, and possibly easier way of describing an algorithm. An example is provided in \Cref{alg:cyclesixcoloring} from~\cite{FraigniaudLR22}, which is aiming at solving $6$-coloring in cycles.
This algorithm uses the function $\imm(s)$, which allows to perform an immediate snapshot (i.e., a write of the current state~$s$ immediately followed by a snapshot of all the states of the neighbors), and uses the function \textsf{return}, which explicitly provides the output (instead of using $\stateterm(x)$).

\algblockdefx[forever]{Forever}{Endforever}%
[1][]{\textbf{repeat forever}}%
{\textbf{end repeat}}
\begin{algorithm}[tb]
\caption{An algorithm for 6-coloring cycles. Code of node $v$, with sole input $\id{v}$.} \label{alg:cyclesixcoloring}
\begin{algorithmic}
\Procedure{CycleSixColoring}{$\id{v}$}
  \State $x \gets \id{v}$; \; $a \gets 0$; \; $b \gets 0$; \Comment{$(x,a,b)$ is the state $s$ of $v$}
  \Forever
     \State $(s_1,s_2) \gets \imm(s)$ \Comment{$s_1$ and $s_2$ are the states of the two neighbors of $v$} 
     \If{$(a,b) \notin \{(s_1.a,s_1.b),(s_2.a,s_2.b)\}$}
         \Return $(a,b)$
     \Else \Comment{In the following: $s_i = \bot \Longrightarrow (s_i.x=\bot) \land (s_i.a=\bot) \land (s_i.b=\bot)$.}
        \State $a \gets \min \mathbb{N} \smallsetminus \{s_i.a \mid (i\in\{1,2\}) \land (s_i \neq \bot) \land  (s_i.x > x)\}$
        \State $b \gets \min \mathbb{N} \smallsetminus \{s_i.b \mid (i\in\{1,2\}) \land (s_i \neq \bot) \land (s_i.x < x)\}$
     \EndIf
  \Endforever
\EndProcedure
\end{algorithmic}
\end{algorithm}

In \Cref{alg:cyclesixcoloring}, the state $s$ of each (non terminated) node is a triplet $s=(x,a,b)$ of natural numbers. Given a state $s$, $s.x$, $s.a$, and $s.b$ respectively denote the first, second, and third element in~$s$. The state of a terminated node is a pair $(a,b)$ of natural numbers. One can check (see~\cite{FraigniaudLR22}) that the output pairs $(a,b)$ can take at most 6 different values.

The state $s$ of a node $v$ is updated by updating some of all of its components $x$, $a$, or~$b$. Actually, the entry $x=\id{v}$ does not change. The entry $a$ is updated to the smallest natural number excluding the $a$-values used by neighbors of larger identifiers, and $b$ is updated to the smallest natural number excluding the $b$-values used by the neighbors of smaller identifiers. These values are equal to~$\bot$ if they have not yet been written in the register (i.e., if a neighbor has not yet performed a single write). 
If a node $v$ notices that its current state $(x,a,b)$ is such that $(a,b)$ is different from the $(a,b)$-pairs of both neighbors, then $v$ terminates, and decides color $(a,b)$. An example of an execution of \Cref{alg:cyclesixcoloring} is provided in Appendix~\ref{app:example:alg:cyclesixcoloring}.

\section{Results and Road Map}

We have now all ingredients sufficient to formally state our results. 

\subsection{Algorithms for ASYNC LOCAL}

We first show (cf. \Cref{sec:linial}) that Linial's $O(\Delta^2)$-coloring algorithm can be adapted to work in the asynchronous wait-free setting.

\begin{restatable}{theorem}{linial}
\label{thm:linial}
For every $\Delta\geq 2$, the round-complexity of $O(\Delta^2)$-coloring graphs of maximum degree $\Delta$ in the \WFLOC model is $O(\log^* n)$.
\end{restatable}

Then, we show (cf. \Cref{sec:savecolors}) that, at the cost of increasing the runtime by an additive factor depending on $\Delta$, it is possible to reduce the number of colors from $O(\Delta^2)$ to $(\Delta+1)(\Delta+2)/2$.

\begin{restatable}{theorem}{savecolors}
\label{thm:savecolors}
For every $\Delta\geq 2$, the round-complexity of $\frac12(\Delta+1)(\Delta+2)$-coloring graphs of maximum degree $\Delta$ in the \WFLOC model is $O(\log^* n)+f(\Delta)$, where $f$ is a function depending on $\Delta$ only.
\end{restatable}

Finally, we show (cf. \Cref{sec:onemore}) that we can exploit the fact that the coloring produced by \Cref{thm:savecolors} satisfies special properties for  reducing the size of the color palette by one color. 

\begin{restatable}{theorem}{generalization}
\label{thm:generalization}
For every $\Delta\geq 2$, the round-complexity of $(\frac12(\Delta+1)(\Delta+2) - 1)$-coloring graphs of maximum degree $\Delta$ in the \WFLOC model is $O(\log^* n)+f(\Delta)$, for some function~$f$ that only depends on $\Delta$.
\end{restatable}

An important consequence of this result is the case $\Delta=2$. Theorem~\ref{thm:generalization} shows that there is an algorithm for $5$-coloring cycles. While such an algorithm was already claimed to exist in~\cite{FraigniaudLR22}, we show (cf.  \Cref{ssec:5cycle}) that the algorithm supporting that claim is erroneous. Specifically, we provide an instance in which the algorithm does not terminate. Theorem~\ref{thm:generalization} provides a novel algorithm, which allows us to establish the following result.

\begin{corollary}\label{cor:avoidingzebug}
The round-complexity of $5$-coloring cycles in the \WFLOC model is $O(\log^* n)$.
\end{corollary}

\subsection{Impossibility Results}

As pointed out in~\cite{FraigniaudLR22} several impossibility results for \WFLOC are mere consequences of the fact the this model coincides with the standard wait-free shared-memory model whenever the underlying graph $G$ is a clique~$K_n$. This is for instance the case of the impossibility of $4$-coloring~$C_3$ (by reduction from renaming), and the impossibility of constructing a maximal independent set, i.e., MIS (by reduction from strong symmetry breaking). Whether or not it is possible to $4$-color cycles $C_n$ for $n>3$ was left open in~\cite{FraigniaudLR22}. We show that, for infinitely many values of~$n$, the problem of $4$-coloring the $n$-node cycle $C_n$ is not solvable in \WFLOC. To establish this result, we prove a result of independent interest, in the framework of wait-free shared memory computing. Specifically, we extend the proof in~\cite{AttiyaP16} that weak symmetry breaking is impossible in the wait-free shared memory systems. We show that this problem remains impossible even if some input are provided to the processes, which may potentially help them to break symmetry. The set of possible inputs has to agree with some restrictions, called \goodInput\/ (with respect to a particular subset of processes). Roughly, the set of possible input assignments must not be divisible by the number~$n$ of processes whenever $n$ is prime, and it must be closed under permuting the identifiers of a particular subset of the processes by an order-invariant permutation. Also recall that  an algorithm is \emph{symmetric} if for every execution $\alpha$ on a subset~$P$ of processes, and for every permutation $\pi:[n]\to [n]$ order preserving on~$P$, we have that, for every $i\in P$, process $i$ outputs $x$ in $\alpha$ if and only if process $\pi(i)$ outputs $x$ on the execution $\pi(\alpha)$ resulting from permuting the scheduling of the processes in~$P$ according to~$\alpha$. Our impossibility results are shown in Section~\ref{sec:lb}.
 
\begin{restatable}{theorem}{wsb}
\label{thm:wsb}
Let $n$ be a prime number. There are no symmetric wait-free deterministic algorithms solving weak symmetry break in the asynchronous wait-free shared memory model with $n$ processes, even if the processes are provided with inputs from a \goodInput\/ set of inputs.
\end{restatable}

\Cref{thm:wsb} has three important consequences. 

\begin{restatable}{corollary}{fourcol}
\label{thm:4col}
Let $n\geq 3$ be a prime number. The problem of $4$-coloring the $n$-node cycle cannot be solved deterministically in \WFLOC. 
\end{restatable}

A weaker form of symmetry breaking is weak $2$-coloring~\cite{NaorS95}. It is required to 2-color the input graph such that every (non isolated) node has at least one neighbor colored with a different color.  

\begin{restatable}{corollary}{weaktwocol}
\label{thm:weaktwocol}
Let $n\geq 3$ be a prime number. The problem of weak $2$-coloring the $n$-node cycle cannot be solved deterministically in \WFLOC.
\end{restatable}

Finally, we prove that, for even values of $\Delta$, there are a infinitely many $\Delta$-regular graphs that cannot be $(\Delta+2)$-colored in \WFLOC. This extends the lower bound of $2\Delta+1$ colors, which applies only for the clique of $\Delta+1$ nodes with $\Delta+1$ power of a prime, to an infinite family of graphs with maximum degree~$\Delta$. 

\begin{restatable}{corollary}{deltaplusthreecolor}
\label{thm:deltaplusthreecolor}
Let $\Delta$ be an even number, and let $n> \Delta$ be a prime number. The problem of $(\Delta+2)$-coloring $n$-node $\Delta$-regular graphs cannot be solved deterministically in \WFLOC.
\end{restatable}

We complete \Cref{sec:lb} with some additional results. The version of MIS considered in~\cite{FraigniaudLR22}, which was proved impossible to solve, asks the nodes to output a set of vertices which forms an MIS in the graph induced by the \emph{correct} nodes. Instead, we consider a weaker variant of MIS, asking the nodes to output a set of vertices which forms an MIS in the graph whenever all processes are correct, i.e., no crashes occurred. We show that even this weaker variant of MIS is impossible in \WFLOC. 

\begin{restatable}{theorem}{weakmis}
\label{thm:weakmis}
For every $n\geq 7$, no deterministic algorithms can solve weak MIS in the $n$-node cycle under \WFLOC.
\end{restatable}

Finally,  we show impossibility results for coloring general graphs (cf \Cref{ssec:lbcol}).

\begin{restatable}{theorem}{lbcol}
\label{thm:lbcol}
For every $\Delta\geq 2$, no deterministic algorithms can solve $(\Delta+1)$-coloring in trees of maximum degree~$\Delta$ under \WFLOC.
\end{restatable}

We conclude, in \Cref{sec:open}, with some open questions.

\section{Coloring General Graphs with $O(\Delta^2)$ Colors}\label{sec:linial}

In this section, we provide a simple algorithm for coloring a graph with $O(\Delta^2)$ colors. This algorithm is an adaptation of Linial's coloring algorithm~\cite{Linial92} (which is designed to work in the \LOCAL model) to the asynchronous setting. More in detail, we prove the following result.

\linial*

In order to prove this result, we start by summarizing Linial's coloring algorithm, and then we show how to adapt it to the wait-free setting. We start by recalling the notion of set systems and of cover-free family of sets.

\begin{definition}
    A \emph{set system} is a pair $(X,\mathcal{F})$, where $X$ is a set, and $\mathcal{F}$ is a collection of subsets of $X$.
A set system $(X,\mathcal{F})$ is a \emph{$k$-cover-free} family if, for every choice of $k+1$ distinct sets $S_0, S_1, \ldots, S_k$ in $\mathcal{F}$, the following holds:  
$
S_0 \smallsetminus  \bigcup_{i=1}^{k} S_i\neq \varnothing.
$
\end{definition}

To provide an intuition about how to use these two definitions, let us assume that the nodes of the input graph $G$ are properly $c$-colored, and let us assume that  there exists a $\Delta$-cover-free family $(X,\mathcal{F})$ satisfying $c \le |\mathcal{F}|$. It follows from these assumptions that there exists a one-to-one function~$f$ from the set of colors to $\mathcal{F}$. W.l.o.g., assume that $X$ contains the numbers in $\{1,\ldots,|X|\}$. One step of Linial's algorithm is able to recolor the nodes with $c' = |X|$ colors, as follows.
\begin{enumerate}
    \item Every node $v$ communicates with its $d$ neighbors to get their current colors $c_1,\ldots,c_d$, where $d \le \Delta$ is the degree of $v$. 
    \item Every node $v$ computes $X_v = f(c_v) \smallsetminus \bigcup_{i=1}^{d} f(c_i)$, where $c_v$ is the color of~$v$, and then recolors itself with the minimum value in $X_v$. 
\end{enumerate}
Note that $X_v$ is guaranteed to be non-empty by the fact that $(X,\mathcal{F})$ is a $\Delta$-cover-free family, and that the obtained color $c'_v$ satisfies $1 \le c'_v \le c'$. Linial's coloring algorithm repeats this process multiple times, each time using a different cover-free family.
The runtime and the resulting number of colors depend on the choice of cover-free families. We summarize the cover-free families used by Linial's algorithm in the following two lemma.

\begin{lemma}[\cite{Linial92}]\label{lem:cf}
    {\rm\textbf{(a)}}~For any $c > \Delta$, there exists a $\Delta$-cover-free family $(X,\mathcal{F})$ with $c \le |\mathcal{F}|$, and $|X| \le 5 \lceil \Delta^2 \log c \rceil$. 
     {\rm\textbf{(b)}}~There exists a $\Delta$-cover-free family $(X,\mathcal{F})$ with $10\Delta^3 \le |\mathcal{F}|$, and $|X| \le (4\Delta+1)^2$.
\end{lemma}

In \cite{Linial92}, \Cref{lem:cf} has been proved in a non-constructive way. However, it is possible to obtain a similar statement by using polynomials over finite fields \cite{coverfree}. We will use the above lemma as a black-box. However, the correctness of our algorithm will be independent from which specific cover-free family construction is used.

We now discuss how these cover-free families are used. Linial's algorithm, in its standard formulation for \LOCAL, requires the nodes to be aware of an upper bound $N$ on the size of the identifier space. At the first round, nodes recolor themselves by using $5 \lceil \Delta^2 \log N \rceil$ colors, thanks to a cover-free family from \Cref{lem:cf}(a) with parameter $c = N$. We denote by $f_1$ the one-to-one function used by the nodes to map their color to the elements of the cover-free family. At the second round, nodes use the cover-free family from \Cref{lem:cf}(a) with parameter $c = 5 \lceil \Delta^2 \log N \rceil$, from which they obtain a coloring that uses $5 \lceil \Delta^2 \log (5 \lceil \Delta^2 \log N \rceil) \rceil$ colors. We denote by $f_2$ the one-to-one function used by the nodes to map their color to the elements of the cover-free family. The nodes repeat this process multiple times, each time using a cover-free family from \Cref{lem:cf}(a) with parameter $c$ equal to the amount of colors obtained in the previous rounds.
Linial proved that it takes $O(\log^* N)$ rounds to reach a coloring that uses at most $10\Delta^3$ colors. Since it is typically assumed that $N = \poly(n)$, the runtime is $O(\log^* n)$.
At this point, the cover-free family from \Cref{lem:cf}(b) is used to get a coloring that uses $(4\Delta+1)^2 = O(\Delta^2)$ colors.

Let us denote by $T$ the number of rounds performed in total, including the last round that uses the family from \Cref{lem:cf}(b) for reducing the number of colors to at most $(4\Delta+1)^2$. For $1 \le i \le T$, let us denote by $f_i$ the one-to-one function used by the nodes to map their colors to the elements of the cover-free family while executing the $i$th round of Linial's algorithm.

\subparagraph{The Algorithm.} 

Let us show that the approach used in Linial's \LOCAL algorithm can be adapted to work in \WFLOC as well. We assume that $\inp{v}$ contains the same upper bound~$N$ on the range of identifiers. So, in particular, every node $v$ can compute~$T$ as a function of $\inp{v}$. The adaptation of Linial's coloring algorithm to \WFLOC is displayed as \Cref{alg:linial}. The main challenge when running Linial's algorithm in the \WFLOC model comes from the fact that a vertex $v$ may be in the $i$th iteration of Linial's algorithm, while a neighbor $u$ of $v$ may be in iteration $j \neq i$. Nevertheless, we will prove that our adaptation of Linial's algorithm correctly handles these cases.  The runtime of \Cref{alg:linial} is clearly $O(\log^* n)$.  We now argue that \Cref{alg:linial} is correct.

\begin{algorithm}[tb]
\caption{$O(\Delta^2)$-coloring arbitrary graph. Code of node~$v$: $\id{v}\in\{1,\dots,N\}$;  $\inp{v}=N$.}
\label{alg:linial}
\begin{algorithmic}[1]
\Procedure{WaitFreeLinial}{$\id{v}$,$\inp{v}$}
  \State $S  \gets (\id{v},\bot,\ldots\bot)$;  \Comment{$S$ is an array of length $T+1=O(\log^\star N)$, and is the state $s$ of $v$}
  \For{$i=1$ \textbf{to} $T$}\label{alglinial:linefor}
      \State $(s_1,\ldots,s_d) \gets \imm(s)$
      \State $A \gets \{ s_j.S[i] \mid (j\in\{1,\dots,d\}) \land (s_j.S[i] \neq \bot) \}$
            \Comment{$i$th entry of each array $s_j.S$}
      \State $S[i+1] \gets \min f_{i}(S[i]) \smallsetminus \bigcup_{a \in A} f_{i}(a) $\label{alglinial:min}
  \EndFor
  \State \Return $s[T+1]$
\EndProcedure
\end{algorithmic}
\end{algorithm}

\subparagraph{Correctness.} 

Let us show that \Cref{alg:linial} produces an $O(\Delta^2)$ coloring. We prove the statement by induction, by proving that, if two neighboring nodes $u$ and $w$ have both executed the $i$th step of the for-loop, then $s_u.S[i+1] \neq s_w.S[i+1]$. For $i=0$, that is, both nodes did not execute any step in the for loop, the statement is clearly satisfied, since $s_u.S[1] = \id{u} \neq \id{w} = s_w.S[1]$. Now, assume $s_u.S[i] \neq s_w.S[i]$. We prove that, under this inductive hypothesis, after $u$ and $w$ have both executed the $i$th iteration of the for-loop, we have  $s_u.S[i+1] \neq s_w.S[i+1]$. 
First of all, at line \ref{alglinial:min}, the set $f_i(S[i]) \setminus \bigcup_{a \in A}f_i(a)$ is not empty (and hence, applying the min operator is well-defined).
This is because nodes are using a $\Delta$-cover-free family, and $s_u.S[i] \neq s_w.S[i]$. Then, there are three cases to consider:
\begin{itemize}
    \item Node $u$ executes step $i$ strictly after node $w$ --- that is, the $i$th $\imm$ operation performed by $u$ is scheduled at a time that is strictly larger than when the $i$th $\imm$ operation performed by $u$ is scheduled. In this case, we are guaranteed that, when $u$ reads the memory of its neighbors, it sees $s_w.S[i] \neq \bot$. Hence, we obtain that $s_w.S[i+1] \in f_{i}(s_w.S[i])$ and that $s_u.S[i+1] \in f_{i}(s_u.S[i]) \setminus f_{i}(s_w.S[i])$, implying $s_u.S[i+1] \neq s_w.S[i+1]$, as required.
    \item Node $u$ executes step $i$ strictly before node $w$. This case is symmetric to the previous one.
    \item Nodes $u$ and $w$ execute step $i$ in parallel --- that is, they perform the $\imm$ operation in parallel. In particular, this implies that when $u$ reads the memory of its neighbors, it sees $s_w.S[i] \neq \bot$, and hence we are in the same situation as in the first case.
\end{itemize}
In all cases, we obtain that $s_u.S[i+1] \neq s_w.S[i+1]$,  which concludes the proof of Theorem~\ref{thm:linial}.  

\section{Reducing the Colors to $(\Delta+1)(\Delta+2)/2$}
\label{sec:savecolors}

In this section, we show that, at the cost of increasing the running time by an additive factor depending on $\Delta$ only, we can decrease the amount of colors from $O(\Delta^2)$ to $\frac12 (\Delta+1)(\Delta+2)$. 

\savecolors*

The algorithm that we provide is a generalization to general graphs of the $6$-coloring algorithm for cycles presented in~\cite{FraigniaudLR22}, and restated in Algorithm~\ref{alg:cyclesixcoloring}.
On a high-level, the algorithms works as follows. First, we compute an initial $O(\Delta^2)$-coloring of the nodes. Then, the final color of each node is given by a pair $(a,b)$. This pair is computed by repeatedly updating the values of $a$ and $b$ until the pair is different from the pairs of the neighbors. The value of $a$ is updated as a function of the $a$-values of the neighbors with larger initial color, while the value of $b$ is updated as a function of the $b$-values of the neighbors with smaller initial color.

\subparagraph{The algorithm.}

In order to prove \Cref{thm:savecolors}, we first analyze the algorithm $\textsc{SaveColors}$, displayed as \Cref{alg:savecolors}. Given an $O(\Delta^2)$-coloring as input, this procedure produces a $((\Delta+1)(\Delta+2)/2)$-coloring, in $f(\Delta)$ rounds for some function $f$.
\Cref{thm:savecolors} follows by running Algorithm \textsc{WaitFreeLinialReduced} below, in which if a node $v$ is running $\textsc{SaveColors}$ while some neighbor $u$ of $v$ is still running $\textsc{WaitFreeLinial}$, then $v$ treats the memory of $u$ as $\bot$.

\begin{algorithmic}
\Procedure{WaitFreeLinialReduced}{$\id{v}$,$\inp{v}$}
 \State $c_v \gets \textsc{WaitFreeLinial}(\id{v},\inp{v})$
 \State $\Return$ $\textsc{SaveColors}(\id{v},c_v)$
\EndProcedure
\end{algorithmic}

\begin{algorithm}[tb]
\caption{Reducing the number of colors from $O(\Delta^2)$ to $(\Delta+1)(\Delta+2)/2$.} \label{alg:savecolors}
\begin{algorithmic}[1]
\Procedure{SaveColors}{$\id{v}$,$\inp{v}$}
    \State $x \gets \inp{v}$; $(a,b) \gets (0,0)$ \Comment{$x\in [O(\Delta^2)]$ is the original color of $v$ }
  \Forever \label{alg2:loop}  \Comment{$s=(x,a,b)$ is the state of $v$}
     \State $(s_1,\ldots,s_d) \gets \imm(s)$ 
     \If{$(a,b) \notin \{(s_i.a,s_i.b) \mid (i\in\{1,\dots,d_v\}) \land (s_i \neq \bot)\}$} \label{alg2:if}
     \Return $(a,b)$
     \Else
        \State $a \gets \min \mathbb{N} 
        \smallsetminus \{s_i.a \mid (i\in\{1,\dots,d_v\}) \land (s_i \neq \bot) \land (x < s_i.x) \}$
        \State $b \gets \min \mathbb{N} 
        \smallsetminus \{s_i.b \mid (i\in\{1,\dots,d_v\}) \land (s_i \neq \bot) \land (x >  s_i.x) \}$
     \EndIf
  \Endforever
\EndProcedure
\end{algorithmic}
\end{algorithm}

\subparagraph{Correctness.}

We start by proving that the  coloring produced by  \Cref{alg:savecolors} is proper, and that the amount of possible resulting pairs is indeed $(\Delta+1)(\Delta+2)/2$. 
First, by assumption, the input of the nodes is a proper coloring, and hence two neighboring nodes have their $x$ variables initialized to different values.
Then, consider two neighboring nodes $u$ and $v$ that terminated with the pairs $(a_u,b_u)$ and $(a_v,b_v)$. There are three cases to consider:
\begin{itemize}
    \item Node $u$ terminated after $v$, that is, the last $\imm$ operation performed by $u$ is scheduled at a time that is strictly larger than the time at which the last $\imm$ operation performed by $u$ is scheduled. In this case, we are guaranteed that, during the last $\imm$ operation performed by node $u$, node $u$ read the pair with which node $v$ terminated, and then $u$ decided to terminate without altering its own pair. This implies that $(a_u,b_u) \neq (a_v,b_v)$, as desired.
    \item Node $u$ terminated before $v$. This case is symmetric to the previous one.
    \item Nodes $u$ and $v$ terminated at the same time. In this case, nodes $u$ and $v$ performed their last $\imm$ operation at the same time, and hence they both read the pair with which the neighbor terminated. Since the condition of the \emph{if} was true, $(a_u,b_u) \neq (a_v,b_v)$, as desired.
\end{itemize}
We now prove a bound on the amount of possible different pairs returned by the algorithm. Consider a node $v$ of degree $d_v$. Observe that $0 \le a,b \le d_v \le \Delta$.
Since each neighbor $i$ of $v$ satisfies $s_i.x > x$ or $s_i.x < x$, but not both, it must be the case that $a + b \le \Delta$.
We thus obtain that the number of possible resulting pairs $(a,b)$ is at most 
$
\sum_{a=0}^{\Delta} (\Delta - a + 1) = (\Delta+1)(\Delta+2)/2,
$
as desired.

\subparagraph{Runtime.}

We prove that the runtime  \Cref{alg:savecolors} is bounded by $f(C,\Delta)$ for some function $f$, where $C = O(\Delta^2)$ is the size of the color palette given as input to $\textsc{SaveColors}$. 
In the following, for every node $v$,  we denote by $v.x$ the variable $x$ in the register of node~$v$. We start by proving a lemma which roughly states that if a node $v$ is scheduled sufficiently many times without having its $b$ variable changed, then $v$ terminates.

\begin{lemma}\label{lem:ub_no_b_change}
    Let $\mathcal{S} = S_1,S_2,\ldots$ be a scheduling. Let $q$ and $r$ be indices such that $q<r$, $v \in S_{q}$, $v \in S_{r}$, and $|\{ j\in\{q,\dots,r\} \mid v \in S_j\}| = \Delta+3$. Assume that for all $j\in\{q,\dots,r\}$ satisfying $v \in S_j$, node $v$ has the same value $v.b$ when performing its $\imm$ operation. Then, there exists a step $t < r$ at which $v$ terminates.
\end{lemma}

\begin{proof}
    Let $q = k_1 < k_2 < \cdots < k_{\Delta+3} = r$
        denote the~$\Delta+3$ steps
        in which node~$v$ is scheduled
        during the interval~$\{q, \ldots, r\}$,
        and let~$I \coloneqq \{k_1, \ldots, k_{\Delta+3}\}$.
    Let $i$ be the smallest value satisfying that, for all $u\in N(v) \cap \bigcup_{j = k_{i-1}+1}^{k_i} S_j$, there exist two indices $j_1$ and $j_2$ satisfying $u \in S_{j_1}$, $u \in S_{j_2}$, and $q \le j_1 < j_2 \le k_i$. That is, every neighbor of~$v$'s that scheduled in between step $k_{i-1} + 1$ and step $k_i$ has been scheduled at least twice between step $q$ and step~$k_i$. We will later show that such an index~$i$ always exists, and that $i\leq \Delta+2$.
    Let us first prove that node~$v$ terminates no later than in step~$k_i$, which is enough to conclude.
    We split the neighbors of $v$ into two groups, and we analyze them separately.
    \begin{itemize}
        \item  \textbf{\boldmath Neighbors $u$ of $v$ that are scheduled at least once in the step interval $[k_{i-1}+1,k_{i}]$.} By assumption, $u$ is scheduled at least twice in the interval $[q,k_{i}]$, say at steps $j_1$ and $j_2$.
        Let us show that, in step~$k_i$,
        the value of~$u.b$ read by node~$v$
        is distinct from~$v.b$,
        so that node~$u$ does not prevent~$v$
        from terminating in this step.
        
        We first consider the case $u.x > v.x$. Node $v$ writes $v.b$ at step $q$.
        Node $u$ then reads~$v.b$ at step $j_1 \ge q$,
        and updates $u.b$ to a new value different from $v.b$.
        Node~$u$ then writes this new value~$u.b \neq v.b$
        in step~$j_2 \le k_i$.
        Since, by assumption,
        $v.b$ is the same at step $q$ as it is at step $k_i$,
        it holds that $u.b \neq v.b$ at step $k_i$.
        For the case~$u.x < v.x$,
            we note by contradiction
            that if we had~$u.b = v.b$ in step~$k_i$,
            node~$v$ would write a new value for~$u_b$ in step~$k_{i+1}$.
        Since~$i \le \Delta+2$,
            this would contradict the assumption
            that node~$v$ writes the same value~$v.b$
            in all steps~$k_1, \ldots, k_{\Delta+1}$.
        Thus~$u.b \neq v.b$ in step~$k_i$.

        \item \textbf{\boldmath Neighbors $u$ of $v$ that are not scheduled in the step interval $[k_{i-1}+1,k_{i}]$.} 
        We consider two cases.
        If node~$u$ has never been scheduled,
            then~$s.u = \bot$ in step~$k_i$
            and node~$v$ ignores~$u$
            in line~\ref{alg2:if} of the algorithm.
        If node~$u$ \emph{has} been scheduled
            sometime before step~$k_{i-1}$,
            node~$v$ reads some pair~$(u.a,u.b)$ in step~$k_{i-1}$,
            updating its own pair~$(v.a, v.b)$
            to avoid colliding with~$u$.
        Since node~$u$ is then not scheduled
            in the interval~$[k_{i-1},k_i]$,
            there is no collision between the pairs
            $(u.a, u.b)$ and~$(v.a, v.b)$ in step~$k_i$.
        In both cases, node~$u$
            does not prevent node~$v$
            from terminating in step~$k_i$.
        
    \end{itemize}
Finally, let us prove that $i \le \Delta+2$. If an index $i \ge 2$ does not satisfy the requirements, then  there is at least one node $u$ that is scheduled during the step interval $[k_{i-1} +1,k_{i}]$ that was never scheduled during the step interval $[q,k_{i-1}]$.
Observe that this can happen for at most $\Delta$ times, because a neighbor $u$ of $v$ can satisfy the condition for a at most one possible value of~$i$. The claim follows from the pigeonhole principle.
\end{proof}

We now use \Cref{lem:ub_no_b_change} to prove a bound on the runtime of the algorithm. We define a partition of the nodes of the input graph $G$ as follows. Let 
$
V_1 = \{ v\in V(G) \mid \forall u \in N(v), \; \inp{v} < \inp{u}\}.
$
That is, $V_1$ is the set of nodes that are local minima w.r.t.\ their input color. Then, for every $i\geq 2$, let 
$
V_i = \{ v\in V(G)  \mid \forall u \in N(v) \smallsetminus \cup_{j=1}^{i-1} V_j, \; \inp{v} < \inp{u}\}
$
That is, $V_i$ contains nodes that are local minima among the nodes that are not in any set $V_1,\ldots,V_{i-1}$.
By construction, each set~$V_i$ is an independent set
    with respect to the graph~$G$,
    and there exists some index~$i^* \le C$
    such that~$V_i = \varnothing \iff i > i^*$.
We prove by finite induction on~$i$
    that the nodes in $V_i$
    terminate in at most $f(i,\Delta)$ steps
    for some function $f$.
For the base case~$i = 1$,
    pick~$v \in V_1$.
As node~$v$ is a local minimum,
    its value~$v.b$ never changes,
    and hence by \Cref{lem:ub_no_b_change} 
    it terminates in at most $f(1,\Delta) = \Delta+3$ steps.
We now consider the inductive case $2 \le i \le i^*$.
Pick a node $v \in V_i$,
    and let $B = N(v) \cap (\cup_{j=1}^{i-1} V_j)$. 
If there exist two indices $q$ and $r$ such that $v$ is scheduled for $\Delta+4$ times in the interval $[q,r]$, and no node in $B$ is scheduled in the same interval $[q,r]$, then, since $v.b$ can just change the first time $v$ is scheduled in this interval, it follows from \Cref{lem:ub_no_b_change} that node $v$ terminates at time at most~$r$.  Hence, for a node $v$  not to terminate, it must be the case that at least one node in $B$ is scheduled at least once, for each of the $\Delta+4$ steps during which  $v$ is scheduled. By induction hypothesis, this can happen for at most $\Delta \cdot f(i-1,\Delta)$ times, merely because $v$ has at most $\Delta$ neighbors in $\cup_{j=1}^{i-1} V_j$, and their runtimes are bounded by $f(i-1,\Delta)$.
This implies that the runtime of $v$ is bounded by 
$
f(i,\Delta) = (\Delta+4) + (\Delta+4)\cdot \Delta \cdot f(i-1,\Delta).
$
Since the colors returned by $\textsc{WaitFreeLinial}$ are bounded by $O(\Delta^2)$, we get that the total runtime  of the algorithm is bounded by $(\Delta^2)^{O(\Delta^2)} = 2^{O(\Delta^2 \log \Delta)}$.

\section{Saving One More Color}\label{sec:onemore}

We now modify \Cref{alg:savecolors}
    in order to save one additional color.
This new algorithm,
    shown in \Cref{alg:alg3},
    allows us to establish the following theorem.

\generalization*

An important consequence of this result is Corollary~\ref{cor:avoidingzebug}, that is, the existence of a 5-coloring algorithm for the cycles in the \WFLOC model. This result is original because the 5-coloring algorithm proposed in~\cite{FraigniaudLR22} has a bug (cf. Appendix~\ref{app:algo-has-a-bug} where we exhibit an instance of 5-coloring $C_4$ for which the algorithm in~\cite{FraigniaudLR22} does not terminate). 

\subsection{Intuition of the algorithm}

We start by providing the high level idea of the algorithm.
The algorithm that we provide is similar to \Cref{alg:savecolors}, and it exploits some special properties of the pairs $(a,b)$ that it produces.
Specifically, we modify \Cref{alg:savecolors}
    such that, if a node outputs the pair~$(\Delta,0)$,
    then none of its neighbors output the pair~$(0,\Delta)$.
In this case, we can identify
    the pairs~$(\Delta,0)$ and~$(0,\Delta)$ as the same color,
    reducing the amount of colors in use by one.
    
Notice that a node that outputs the pair~$(\Delta,0)$
    is necessarily a local minimum
    with respect to the node identifiers,
    and similarly a node that outputs the pair~$(0,\Delta)$
    is necessarily a local maximum.
The problematic case of neighbors
    outputting both pairs~$(\Delta,0)$ and~$(0,\Delta)$
    can therefore only happen
    when the both neighbors are local extrema.
However, for such neighboring nodes to reach a state
    where they would output problematic pairs,
    some specific conditions must hold
    which can be handled by the nodes
    as a specific case.
    
More in detail, 
    in such a situation,
    we make nodes flip their relative ordering:
    if node~$u$ is a local minimum
    and node~$v$ is a local maximum,
    then $u$ will treat $v$ as smaller
    when comparing their $x$ variables,
    and $v$ will treat $u$ as larger.
By flipping relative ordering,
    we are forcing neighboring local extrema
    with pairs $(\Delta,0)$ and $(0,\Delta)$
    to stop being local extrema,
    leading them to change their output pairs.
This modification will affect the termination time,
    and hence we also need to introduce new terminating conditions.

\subsection{Formal Description}

The algorithm is displayed in \Cref{alg:alg3}. It refers to some functions that are presented below.

\subparagraph{Treating special pairs as equal.}

The first modification applied to \Cref{alg:savecolors} is the following. In line \ref{alg2:if}, instead of directly using the pairs $(a,b)$ of the node, and the pairs of its neighbors, we first map them by using the function $\textsc{Map}$ shown below. Observe that $\textsc{Map}$ behaves as the identity function for all pairs different from $(\Delta,0)$, and it maps $(\Delta,0)$ to $(0,\Delta)$. In this way, the algorithm behaves similarly as the original one, except that it forbids neighboring nodes with pairs $(0,\Delta)$ and $(\Delta,0)$ to terminate, since after applying $\textsc{Map}$, they are both mapped to $(0,\Delta)$, and hence they are treated as having the same pair.

\begin{algorithmic}
\Procedure{Map}{$a$,$b$} 
    \If{ $(a,b) = (\Delta,0)$ } \Return $(0,\Delta)$
    \Else{} \Return $(a,b)$
    \EndIf
\EndProcedure
\end{algorithmic}

\subparagraph{A new ordering relation.}

In \Cref{alg:savecolors}, nodes exploit their variables $x$ (that is, the given coloring) to determine an ordering relation between them. In the new algorithm, each node keeps an additional variable~$f$, which is a set of identifiers. The semantic is the following. For two nodes $u$ and~$v$, if $u \in v.f$ or $v \in u.f$, then the ordering w.r.t.\ their variables $x$ is flipped. We call an edge $\{u,v\}$ flipped whenever $u \in v.f$ or $v \in u.f$.

Let us define two auxiliary Boolean functions that are used by a node $v$ to determine whether the ordering relation with a neighbor~$u$ should be considered flipped or not. These functions take as input the state $s_v$ and $s_u$ of the two (neighboring) nodes. The variable $z$, as will be shown in the algorithm, stores the identifier of the node.

\begin{algorithmic}
\Procedure{IsNotFlipped}{$s_v$, $s_u$} 
    \State \Return $(s_v \neq \bot) \land (s_u \neq \bot) \land (s_u.z \notin s_v.f) \land (s_v.z \notin s_u.f)$
\EndProcedure
\Procedure{IsFlipped}{$s_v$, $s_u$} 
    \State \Return $(s_v \neq \bot) \land (s_u \neq \bot) \land \big ((s_u.z \in s_v.f) \lor (s_v.z \in s_u.f)\big)$
\EndProcedure
\end{algorithmic}

We are now ready to define the new ordering relation. For this purpose, we define two functions that, given the state $s$ of the node, and the state $s_i$ of its $i$th neighbors, return the neighbors that are considered smaller, and the neighbors that are considered larger, respectively.

\begin{algorithmic}
\Procedure{Smaller}{$s$,$(s_1,\ldots,s_k)$} 
    \State \Return ${\big \{ i \in\{1,\dots,k\} \mid  \big (( \textsc{IsNotFlipped}(s,s_i) \land  (s.x > s_i.x)\big) \lor \big( \textsc{IsFlipped}(s,s_i) \land (s.x < s_i.x)\big)\big \}}$
\EndProcedure
\Procedure{Larger}{$s$,$(s_1,\ldots,s_k)$} 
    \State \Return ${\big \{ i \in\{1,\dots,k\} \mid \big (\textsc{IsNotFlipped}(s,s_i) \land (s.x < s_i.x)\big) \lor \big( \textsc{IsFlipped}(s,s_i) \land (s.x > s_i.x)\big )\big \}}$
\EndProcedure
\end{algorithmic}

\subparagraph{Special termination.}

We also define a function that provides an extra termination condition. It relies on an additional function that detects a neighborhood with special properties. It uses some variables $\alpha$ and $\beta$ that are both set to true if a node has at least one smaller neighbor (that is, it is not a local minima), and it has at least one larger neighbor (that is, it is not a local maxima). We assume that the maximum degree~$\Delta$ is part of the input provided to the nodes.

\begin{algorithmic}
\Procedure{SpecialNeighborhood}{$s$,$(s_1,\ldots,s_\Delta)$}
    \State \Return $\Big (
    \big (\bigwedge_{i=1}^{\Delta} (s_i \neq \bot)\big ) 
    \land 
    \big( \{s.a,s.b\} \cup (\cup_{i=1}^\Delta \{s_i.a,s_i.b\}) \subseteq \{0,\ldots,\Delta-1\} \big)
    \land 
    \; s.\alpha \; \land \; s.\beta  $
    \State \hfill $\land \; \Big ( \bigwedge_{i=1}^{\Delta} \big (
        (s_i.\alpha \lor |\textsc{Smaller}(s_i,[s])| = 1) 
        \land 
        (s_i.\beta \lor |\textsc{Larger}(s_i,[s])| = 1) \big )\Big)\Big)$
\EndProcedure
\end{algorithmic}

That is, a neighborhood of a node $v$ is \emph{special} if (1)~node $v$ has seen all its neighbors, (2)~they are precisely $\Delta$, (3)~the $a$ and $b$ variables of the node and of all its neighbors are in $\{0,\ldots,\Delta-1\}$, and (4)~node $v$ and all its neighbors have at least one smaller, and at least one larger neighbor.
The reason why we use the condition $s_i.\alpha \lor |\textsc{Smaller}(s_i,[s])| = 1$ for checking whether a node has at least one smaller neighbor, instead of just using $s_i.\alpha$ is the following. Let $u$ be the node with state $s_i$, and $v$ be the node with state~$s$. It could be the case that $v$ is smaller than $u$, but $u$ has been scheduled earlier than $v$. So it may be the case that $u$ has never seen $v$. In this case, we could get that $u.\alpha$ is false, even though $u$ has $v$ as smaller neighbor. For this reason, node $v$ computes whether $u.\alpha$ would become true if $u$ were to be scheduled one additional round, by checking whether $v$ is smaller than $u$ using the condition $|\textsc{Smaller}(s_i,[s])| = 1$. A similar reasoning is applied for checking whether a node has at least one larger neighbor. Note that, in the algorithm, once a node sets $\alpha$ (resp.\ $\beta$) to true, that is when it realizes that it is not a local minima (resp., maxima), it will never change its value. The reason is that, as we will prove later, a node never becomes a local minima (resp., maxima) by flipping edges. 
We now introduce the special termination condition. According to this special condition, a node terminates if (1)~its neighborhood is special, and (2)~it is a local maxima according to the original ordering, that is, before flipping any edge.

\begin{algorithmic}
\Procedure{SpecialTermination}{$s$,$(s_1,\ldots,s_\Delta)$}
    \State \Return $\textsc{SpecialNeighborhood}(s,(s_1,\ldots,s_\Delta)) \land \big(\forall i\in \{1,\dots,\Delta\}, s.x > s_i.x\big )$
\EndProcedure
\end{algorithmic}

\subparagraph{The new algorithm.}

The algorithm is displayed as  \Cref{alg:alg3}. 
Like in the case of \Cref{alg:savecolors}, we assume that $\inp{v}$ is the result of running $\textsc{WaitFreeLinial}$.  Observe that the algorithm is similar to \Cref{alg:savecolors}, with only three exceptions. First, it identifies $(0,\Delta)$ with $(\Delta,0)$ when checking for termination at line \ref{alg3:if}. Second,  it uses the custom ordering relation induced by the functions $\textsc{Smaller}$ and $\textsc{Larger}$ at lines \ref{alg3:update_a} and \ref{alg3:update_b}. Third, it has an additional termination condition at line \ref{alg3:if2}.

\algnewcommand{\IIf}[1]{\State\algorithmicif\ #1\ \algorithmicthen}
\algnewcommand{\ElseIIf}[1]{\algorithmicelse\ #1}
\algnewcommand{\EndIIf}{\unskip\ \algorithmicend\ \algorithmicif}

\begin{algorithm}[h]
\caption{Saving 1 color from palette $[\frac12(\Delta+1)(\Delta+2)]$. Algorithm of node~$v$ with color $\inp{v}$} \label{alg:alg3}
\begin{algorithmic}[1]
\Procedure{SaveOneMoreColor}{$\id{v}$,$\inp{v}$}
  \State $a \gets 0$; $b \gets 0$ ; $x \gets \inp{v}$ ; $z \gets \id{v}$
\State $f \gets \{\}$; $\alpha \gets \textbf{false}$; $\beta \gets \textbf{false}$ \Comment{$s=(a,b,x,f,\alpha,\beta,z)$ is the state of node $v$}

  \Forever \label{alg3:loop}
     \State $(s_1,\ldots,s_\Delta) \gets \imm(s)$ \Comment{if $d_v < \Delta$, we assume that $s_i = \bot$ for all $i > d_v$}
     \If{$\textsc{Map}(a,b) \notin \{\textsc{Map}(s_i.a,s_i.b) \mid i\in\{1,\dots,\Delta\} \land s_i \neq \bot\}$} \label{alg3:if}
         \Return $\textsc{Map}(a,b)$ \label{alg3:ret1}
     \Else
        \If{$(a = \Delta) \lor (b = \Delta)$}\label{alg3:flip1} \Comment{We compute the flipped edges.}
            \State $f \gets f \cup \big \{ s_i.z \mid (i\in\{1,\dots,\Delta\}) \and (s_i \neq \bot) \land \big ((s_i.a = \Delta) \lor (s_i.b = \Delta)\big )\big  \}$\label{alg3:flip2}
        \EndIf
        \State $a \gets \mathbb{N} \smallsetminus \{ s_i.a \mid i \in \textsc{Larger}(s,(s_1,\ldots,s_\Delta)) \}$\label{alg3:update_a} 
        \State $b \gets \mathbb{N} \smallsetminus \{ s_i.b \mid i \in \textsc{Smaller}(s,(s_1,\ldots,s_\Delta)) \}$\label{alg3:update_b}
    
        \If{ $|\textsc{Smaller}(s,(s_1,\ldots,s_\Delta))| \ge 1$ }
                $\alpha \gets \textbf{true}$
        \EndIf
        \If{ $|\textsc{Larger}(s,(s_1,\ldots,s_\Delta))| \ge 1$ }
                $\beta \gets \textbf{true}$
        \EndIf
        \If{  $\textsc{SpecialTermination}(s,(s_1,\ldots,s_\Delta))$ }\label{alg3:if2}
             \Return $(0,\Delta)$ \label{alg3:ret2}
        \EndIf
     \EndIf
  \Endforever
\EndProcedure 
\end{algorithmic}
\end{algorithm}
In the rest of the section, we prove that Algorithm~\ref{alg:alg3} is correct.

\subsection{Properties of flipped edges.}

We start by proving that, for an edge to be flipped, some special conditions must be satisfied. In the following, by \emph{local minimum w.r.t.\ $x$} we denote a node that has its $x$ variable smaller than all the $x$ variables of its neighbors. Analogously, we define the notions of local maximum and local extremum w.r.t.~$x$.

\begin{lemma}\label{lem:who_flips}
An edge  $\{u,v\}$ can be flipped only if $u$ and $v$ have degree exactly $\Delta$, and they are both local extrema w.r.t.\ $x$ (one maximum and one minimum).
\end{lemma}

\begin{proof}
    We first prove that $u$ and $v$ need to have degree exactly $\Delta$. 
    Assume, w.l.o.g., that $v$ has degree smaller than $\Delta$. Then, both its $a$ and $b$ variables are in $\{0,\ldots,\Delta-1\}$, and hence the condition in line \ref{alg3:flip1} is not satisfied for $v$, and $u$ will never consider $v$ in line \ref{alg3:flip2}.
    Now, assume that $u$ and $v$ have both degree $\Delta$, and at least one of them, say~$v$, is not a local extremum w.r.t.\ $x$.
    Then, similarly as before, both $v.a$ and $v.b$ are in $\{0,\ldots,\Delta-1\}$. Hence the condition in line \ref{alg3:flip1} is not satisfied for $v$, and $u$ will never consider $v$ in line \ref{alg3:flip2}.
\end{proof}

While \Cref{lem:who_flips} provides a necessary condition for an edge to be flipped, we now prove that an even stronger condition must hold.

\begin{lemma}\label{lem:flipping}
    Let $v$ be a node, and let $u_1,\ldots,u_\Delta$ be its neighbors. Assume that an edge $\{u_i,v\}$ is flipped by either $u_i$ or $v$, for some $i\in\{1,\dots,\Delta\}$. Let $T$ be the first time step in which either $u_i$ performed a $\imm$ operation in which $u_i.f$ contains the $\id{v}$, or $v$ performed a $\imm$ operation in which $v.f$ contains the $\id{u_i}$. Let $T'' > T' \ge T$ be two time steps at which $v$ is scheduled (that is, $v$ performed two $\imm$ operations after the flip).
    Let $s_1,\ldots,s_\Delta$ be the states of the neighbors of $v$ as read by $v$ during the $\imm$ operation of $v$ performed at step $T''$, and let $s$ be the state of $v$. The following holds. 
    \begin{enumerate}
        \item Node $v$ is the only  neighbor of $u_i$ that is local extremum w.r.t.\ $x$;
        \item Node $u_i$ is the only neighbor of $v$ thath is local extremum w.r.t.\ $x$;
        \item Node $v$ has seen all neighbors, and there are $\Delta$ of them (that is, $\bigwedge_{j=1}^{\Delta} (s_j \neq \bot)$);
        \item The $a$ and $b$ variables of all nodes in $\{u_1,\ldots,u_\Delta\}\smallsetminus\{u_i\}$ are in $\{0,\ldots,\Delta-1\}$, that is, \\ $\cup_{j\in\{1,\dots,\Delta\}\smallsetminus\{i\}} \{s_j.a,s_j.b\}\subseteq \{0,\ldots,\Delta-1\}$;
        \item The boolean variables $s.\alpha$ and $s.\beta$ are both true;
        \item $\bigwedge_{j=1}^{\Delta} \big (( (s_j.\alpha) \lor (|\textsc{Smaller}(s_j,[s])| = 1)) \land  ((s_j.\beta)\lor (|\textsc{Larger}(s_j,[s])| = 1))\big )=\mathsf{true}$.
    \end{enumerate}
\end{lemma}

\begin{proof}
    By \Cref{lem:who_flips}, if $\{u,v\}$ is flipped, then  $u$ and $v$ have degree $\Delta$, and they are both local extrema w.r.t.~$x$. W.l.o.g., assume that $v$ is a local minimum, and that $u$ is a local maximum w.r.t.~$x$. 
    Observe that $a+b \le \Delta$ is always be satisfied. Hence, for the flip to happen, it must hold that $(v.a,v.b) = (\Delta,0)$ and $(u.a,u,b) = (0,\Delta)$. Let $t$ be the first time at which $v$ performs a $\imm$ operation in which $v.a = \Delta$. Observe that point~(3) must be satisfied.
    By the definition of~$t$, during time steps $1,\ldots,t-1$, it holds that $v.a \neq \Delta$. Observe that $v.b$ is initialized to $0$, and by an inductive argument, since $v$ is a local minima, we get that in all time steps in $1,\ldots,t-1$, no edges are flipped and hence $v.b$ stays $0$.
    
    Let $t'$ be the last time in which $v$ is scheduled before step $t$, that is, the time in which $a$ has been updated to $\Delta$. At time $t'$, $v$ performed a $\imm$ operation. It saw all its neighbors, and they must have their values of $a$ all different in $\{0,\ldots,\Delta-1\}$. This is because otherwise $v$ would have not set $a$ to $\Delta$. Consider all neighbors $W$ of $v$ that  have $a > 0$ at time $t'$. 
    
    According to the ordering relation at time $t'$, all nodes in $W$ are not local maxima, and they have $\beta$ set to true. Let us show that all nodes in $W$ have never been local maxima before step~$t'$, implying that they are not local maxima w.r.t.\ $x$. 
    
    For a node $w$ to be a local maximum w.r.t.\ $x$, and not being a local maximum at time $t'$, it must hold that an edge incident to $w$ has been flipped. This edge cannot be $\{w,v\}$ because, before time $t'$, $v.a \neq \Delta$ and $v.b = 0$. So it must be some other edge $\{w,w'\}$ with $w' \neq v$. For this edge to be flipped, by \Cref{lem:who_flips}, it must hold that $w'$ is a local minimum w.r.t.\ $x$, and $w$ is a local maximum w.r.t.\ $x$. We consider two possible cases. 
    
    - If $w'$ has no other incident edges flipped up to time $t'$, then $w'.b = 0$ at any point in time, implying that $w$ would have never set $w.b = \Delta$, since it has two neighbors with $b =0$ (that is, nodes $w'$ and $v$), reaching a contradiction. 
    
    - If $w'$ has some other edge that is flipped, then $w'$ got one smaller neighbor and updated $a$ to a value that is at most $\Delta - 1$, and hence $\{w,w'\}$ would not be flipped, reaching a contradiction. 
    
    We obtain that all nodes in $W$ are not local maxima w.r.t.\ $x$. Note that they cannot be local minima w.r.t.\ $x$ either, because they are neighbors of $v$, which is a local minimum w.r.t.\ $x$. Therefore, we obtain that nodes in $W$ are not local extrema w.r.t.\ $x$, proving points~(1) and~(2). We now prove that all the other conditions are also satisfied. 
    
    Since all nodes in $W$ are not local extrema w.r.t.\ $x$, and by \Cref{lem:who_flips}, they will never be local extrema, and will have their variables $a$ and $b$ in $\{0,\ldots,\Delta-1\}$. Hence, point~(4) is satisfied.
    
    We proved that all nodes in $W$ have their $\beta$ variable set to true. Moreover, for every node $w \in W$, it holds that $|\textsc{Smaller}(s_w,[s_v])| = 1)$.  Hence, point~(6) is satisfied for all nodes in $W$.
    
    Since $u_i$ is a local maximum w.r.t.\ $x$, and since it must have seen all its neighbors, we get that  $u_i.\alpha=\mathsf{true}$.  Since $v$ is a local minimum, and since it must have seen all neighbors, we also get that  $v.\beta=\mathsf{true}$. 
    At time $T$, the flip is written by $u_i$ or~$v$. At time at most~$T'$, $v$ is aware of the flip. Hence, according to the new ordering, $v$ knows that it has a larger neighbor (node~$u_i$), and hence $v$ sets $v.\alpha$ to true. This value is written at time $T''$. Also, according to the new ordering, $|\textsc{Larger}(s_{u_i},s_v)| = 1$. Hence, point~(5) is satisfied, and point~(6) is satisfied too, for $u_i$.
\end{proof}

\Cref{lem:flipping} implies the following two corollaries that we will use later.

\begin{corollary}\label{cor:nodes_stay_mid}
     If a node $v$ has at least one smaller (resp., larger) neighbor according to~$x$, then it never happens that $v$ stops with at least one smaller (resp., larger) neighbor, even after possible flipping of edges.
\end{corollary}

\begin{proof}
    By point~(1) and~(2) of \Cref{lem:flipping}, for an edge $\{u,v\}$ to be flipped, $u$ and $v$ must be local extrema w.r.t.~$x$, and they have no other neighbors that are local extrema w.r.t.~$x$. 
\end{proof}

\begin{corollary}\label{lem:oneflip}
    Assume node $v$ satisfies the condition at line \ref{alg2:if} of \Cref{alg:savecolors}, but not the condition at line \ref{alg3:if} of \Cref{alg:alg3}. Then, at the time in which~$v$ executes line \ref{alg3:if} of \Cref{alg:alg3}, we have $v.f=\varnothing$.
\end{corollary}

\begin{proof}
    If the premise holds, then $v$ flips one of its incident edges, say edge $\{u,v\}$. We now apply \Cref{lem:flipping} on the edge $\{u,v\}$, and we obtain that $u$ is the only neighbor of $v$ that is a local extremum w.r.t.~$x$. By applying \Cref{lem:who_flips}, we obtain that, for every node $u_i \in N(v)$ distinct from~$u$, $u_i.z \notin v.f$. Moreover, if $u \in v.f$, then $v.a$ and $v.b$ are both in $\{0,\ldots,\Delta-1\}$, contradicting the premise.
\end{proof}

\subsection{Proof of Correctness} 

Before showing that our algorithm is correct, we first prove a property that is useful for showing correctness. More precisely, we show that, once a neighborhood satisfies the property of being special, it can never happen that, after additional scheduling steps, the neighborhood is not special anymore.

\begin{lemma}\label{lem:standard}
    A special neighborhood $(v,\{u_1,\ldots,u_\Delta\})$ never becomes non-special.
\end{lemma}

\begin{proof}
    Let $t$ a time in which a neighborhood becomes special. We prove by induction that the neighborhood is special at step $t+i$ for any $i$. The base case $i=0$ holds by assumption. Assume that the neighborhood is special at step $t+i$.
    By the definition of special neighborhood, and by \Cref{cor:nodes_stay_mid}, it must hold that, at time $t+i$, all nodes in $U = \{v,u_1,\ldots,u_\Delta\}$ satisfy the property that they have at least one smaller neighbor, and at least one larger neighbor, according to the ordering induced by the flipping status of the edges at time $t+i$. This implies that the nodes in $U$ have the $a$ and $b$ variables in $\{0,\ldots,\Delta-1\}$ at step $t+i+1$. Indeed, in order to have $a = \Delta$ (resp., $b = \Delta$), a node needs to have $\Delta$ larger (res., smaller) neighbors at step $t+i$, which is false by assumption. Finally, by \Cref{cor:nodes_stay_mid}, nodes in $U$ have at least one smaller  neighbor, and at least one larger neighbor at step $t+i+1$. Thus, the neighborhood remains special. 
\end{proof}

We now prove the correctness of our algorithm.
Observe that the pair $(\Delta,0)$ is never returned, and hence the amount of colors is exactly one less than \Cref{alg:savecolors}. We now prove that, if two neighboring nodes terminate, then they obtain different pairs. There are two possible cases where the algorithm terminates, at line \ref{alg3:ret1}, and line \ref{alg3:ret2}. If two neighboring nodes both terminate in line \ref{alg3:ret1}, then their pairs are different for the same reason as in the case of \Cref{alg:savecolors}. Moreover, two neighboring nodes cannot both terminate at line \ref{alg3:ret2}. Indeed, by the definition of the function $\textsc{SpecialTermination}$, for a node to terminate at line \ref{alg3:ret2}, it must be a local maximum w.r.t.~$x$, and two neighboring nodes cannot both be local maxima w.r.t.~$x$.
If a node $v$ terminates at line \ref{alg3:ret2}, and a neighbor $u$ of $v$ terminates at line \ref{alg3:ret1}, then, by \Cref{lem:standard}, $v$ will always be part of a special neighborhood, implying that $\{u.a,u.b\} \subseteq \{0,\ldots,\Delta-1\}$, from which it follows that $u$ does not terminate with $(0,\Delta)$ at line \ref{alg3:ret1}.

\subparagraph{Runtime.}

In the remaining of the section we prove a bound on the runtime of the algorithm. We start by proving a lemma, which is an adaptation of \Cref{lem:ub_no_b_change} to the new algorithm.

\begin{lemma}\label{lem:ub_no_b_change_2}
    Let $\mathcal{S} = S_1,S_2,\ldots$ be a scheduling. Let $r>q$  be indices such that $v \in S_{q}$, $v \in S_{r}$, and $|\{ j\in\{q,\dots,r\} \mid v \in S_j\}| = \Delta+3$. Assume that:
    \begin{enumerate}
        \item For all $j\in\{q,\dots,r\}$ such that $v \in S_j$, node $v$ has the same value $s.b$ when performing the $\imm$ operation at step~$j$. 
        
        \item For all $j\in\{q,\dots,r\}$, node $v$ has the same value $s.f$ when performing the $\imm$ operation at step~$j$. That is, $v$ does not flip additional edges in this time interval.
        
        \item For every node $u\in N(v)$ one of the two cases below holds:
        \begin{itemize}
            \item $v\in u.f$ at all the $\imm$ operations performed in the interval $[q,r]$; 
            
            \item $v\notin u.f$ at all the $\imm$ operations performed in the interval $[q,r]$.  
        \end{itemize}
        That is, the neighbors of $v$ do not flip additional edges incident to $v$ in the time interval $[q,r]$.
    \end{enumerate}
    Then, there exists $t \le r$ such that $v$ terminates at step~$t$.
\end{lemma}

\begin{proof}
    Observe that if, in some time step in the interval $[q,r]$, the condition at line \ref{alg2:if} of \Cref{alg:savecolors} holds while the condition at line \ref{alg3:if} of \Cref{alg:alg3} does not hold, then $v$ immediately decides to flip an edge. By \Cref{lem:oneflip} that edge is not already flipped, and hence point~(2) of the premise of the lemma would not be satisfied.
    Hence, if the premise of the lemma holds, then the terminating condition at line \ref{alg2:if} of \Cref{alg:savecolors} always behaves in the same way as the condition at line \ref{alg3:if} of \Cref{alg:alg3}. 
    
    Moreover, the variables $a$ and $b$ are updated with the same rules as in \Cref{alg:savecolors},  except from the fact that a different ordering relation is used. Note that, by points~(2) and~(3) in the premise, the ordering relation between $v$ and its neighbors does not change in the time interval $[q,r]$. Hence, the proof follows in the same way as in \Cref{lem:ub_no_b_change}.
\end{proof}

We will use \Cref{lem:ub_no_b_change_2} to prove (by induction) that nodes terminate in $f(\Delta)$ rounds for some function~$f$. However, we need to handle nodes with incident flipped edges in a special way. Recall that, by \Cref{lem:who_flips}, these nodes can only be local extrema.

\begin{lemma}\label{lem:max_min_term}
    Let $v$ be a node with degree $\Delta$ that is a local maxima, or a local minima w.r.t.~$x$. Then $v$ terminates in $O(\Delta)$ rounds.
\end{lemma}

\begin{proof}
Let us assume that there is a time interval $[q,r]$ during which $v$ is scheduled for $\Delta+3$ times such that, during  the whole interval, it holds that $v.f=\varnothing$, and $v \notin u.f$ for all $u \in N(v)$.
Observe that points~(2) and~(3) in the premise of \Cref{lem:ub_no_b_change_2} are satisfied.
If $v$ is a local minimum w.r.t.~$x$, then, since $v$ has no flipped edges, $v$~never updates $v.b$, and thus the condition of \Cref{lem:ub_no_b_change_2} applies, from which it follows that $v$ terminates in at most $\Delta+3$ steps. If $v$ is a local maximum then $v$ never updates $v.a$, and thus a symmetric condition to the one of \Cref{lem:ub_no_b_change_2} holds (w.r.t.~$v.a$ instead of w.r.t.~$v.b$), thus $v$ terminates in at most $\Delta+3$ steps.
    
The remaining case to consider is when $v$ becomes incident to flipped edges within its first $\Delta+3$ steps. Let us consider the case in which $v$ is a local maximum w.r.t.~$x$. By \Cref{lem:flipping}, $v$~can only be incident to a single flipped edge, say $\{u,v\}$, and the only neighbor of $v$ that is larger than $v$ w.r.t.\ the ordering after the flipping is $u$.
There are two cases to consider: either $v$ performs $\Delta+3$ steps without $u$ taking any step, or $u$ runs at least once. 

In the former case, $v$ never updates $v.a$, and hence it terminates by \Cref{lem:ub_no_b_change_2}. 

In the latter case, we obtain that both $v$ and $u$ update $a$ and $b$ to values that are at most $\Delta-1$. Moreover $v$ sets $\alpha$ and $\beta$ to true. 
Observe that, by \Cref{lem:flipping}, the condition of $\textsc{SpecialTermination}$ is now satisfied on $v$, and hence it terminates in at most two additional time steps. We obtain that, if $v$ is a local maximum, and $u$ is scheduled at least once, then $v$ terminates within $O(1)$ rounds after $u$ is scheduled.

Consider now the case in which $v$ is a local minimum w.r.t.~$x$. If it is scheduled for $\Delta+3$ steps without $u$ being scheduled, then it terminates by \Cref{lem:ub_no_b_change_2}, merely because $v.b$ is never updated. 
However, as already discussed, $u$ can be scheduled at most $O(1)$ times before terminating. Therefore, after $O(\Delta)$ steps of $v$ there must be $\Delta+3$ steps in which $u$ is not scheduled, from which it follows that $v$ terminates thanks to \Cref{lem:ub_no_b_change_2}.
\end{proof}

We now combine \Cref{lem:ub_no_b_change_2} and \Cref{lem:max_min_term} to prove a bound on the runtime of the algorithm.
We define a partition of nodes as follows. Let 
\[
V_1 = \{ v \mid \forall u \in N(v), \; \inp{v} < \inp{u} \}. 
\]
That is, $V_1$ contains the nodes that are local minima w.r.t.\ their input colors. For $i>1$, let 
\[
V_i = \{ v \mid \forall u \in N(v) \smallsetminus \cup_{j=1}^{i-1} V_j, \; \inp{v} < \inp{u}\}.
\]
That is, $V_i$ contains the nodes that are local minima among the nodes not in $V_1,\ldots,V_{i-1}$. Let $C$ be the amount of colors returned by $\textsc{WaitFreeLinial}$. Observe that, for all $i > C$, $V_i = \emptyset$. Moreover,  for all $i\geq 1$, $V_i$~is an independent set.  We prove by induction on~$i$ that the nodes in $V_i$ terminate in at most $f(i,\Delta)$ steps for some function $f$.

For the nodes in $V_1$, if none of their incident edges get flipped within $\Delta+3$ steps, then, by \Cref{lem:ub_no_b_change_2}, they terminate in at most $\Delta+3$ steps. Otherwise, if they flip an edge, then, by \Cref{lem:who_flips}, they must either be a local maxima, or a local minima. Thus, by \Cref{lem:max_min_term}, they terminate in $O(\Delta)$ steps. Hence, $f(1,\Delta) = O(\Delta)$.

We now consider the case $2 \le i \le C$. Let us consider a node $v \in V_i$, and let $B = N(v) \cap (\cup_{j=1}^{i-1} V_j)$. If $v$ gets an incident flipped edge, then, by \Cref{lem:who_flips}, it is a local extrema, and thus, by \Cref{lem:max_min_term}, $v$ terminates in $O(\Delta)$ additional steps. Thus, in the following, we assume that $v$ never gets an incident flipped edge. 

If there exist two indices $q$ and $r$ such that $v$ is scheduled for $\Delta+4$ times in the interval $[q,r]$, and no node in $B$ is scheduled in the interval $[q,r]$, then, since $v.b$ can only change the first time $v$ is scheduled in this interval, we get from \Cref{lem:ub_no_b_change_2} that node $v$ terminates at time at most~$r$. Therefore, for node $v$  not to terminate, at least one node in $B$ needs to be scheduled at least once every $\Delta+4$ steps of~$v$. By the inductive hypothesis, this can happen for at most $\Delta \cdot f(i-1,\Delta)$ times because $v$ has at most $\Delta$ neighbors in $\cup_{j=1}^{i-1} V_j$, and the runtime of these neighbors is bounded by $f(i-1,\Delta)$. This implies that the runtime of $v$ is itself bounded by 
\[
f(i,\Delta) = (\Delta+4) + (\Delta+4)\cdot \Delta \cdot f(i-1,\Delta).
\]
Since the colors returned by $\textsc{WaitFreeLinial}$ are bounded by $O(\Delta^2)$, we obtain that the runtime is bounded by  $(\Delta^2)^{O(\Delta^2)} = 2^{O(\Delta^2 \log \Delta)}$. This completes the proof of Theorem~\ref{thm:generalization}.

\section{Impossibility Results}\label{sec:lb}

In this section, we will establish several impossibility results for the \WFLOC model. In particular, we will prove that it is impossible to compute a 4-coloring in infinitely many cycles. To establish this specific result, we extend an impossibility result in the standard shared-memory model regarding solving weak symmetry-breaking. 

\subsection{Preliminaries}

Our goal is to simulate algorithms designed for \WFLOC model on the standard shared-memory model, in which $n$ asynchronous crash-prone processes exchange information via single-writer multiple-reader registers. In this model, the $n$ processes have distinct identifiers in $\{1,\dots,n\}$, where $n$ is known to all processes. The major difference between this model and \WFLOC is the absence of graph structure in the shared-memory model. That is, at each activation, a process has access to the registers of all the other processes.  To simulate on a shared-memory system an algorithm designed for the cycle in \WFLOC, we need that each process knows (as an input) which processes are its two neighbors in the cycle. Given this information the simulation is straightforward. However, by providing additional knowledge to the nodes, one cannot use impossibility results for shared-memory blindly, as the additional knowledge may allow the nodes to solve problems that were impossible without this knowledge. Our goal is to extend the proof of impossibility of weak symmetry-breaking for the standard shared-memory model in~\cite{AttiyaP16} to the case in which processes are provided with additional inputs. 

Let us first recall basic concepts used in~\cite{AttiyaP16}.
Let us fix an algorithm $\mathcal{A}$ (later, we shall focus on algorithms solving different versions of symmetry-breaking).
We consider schedulings $\mathcal{S}=S_1,S_2,\ldots$ where every process terminates in~$\mathcal{A}$.
Moreover, we assume that a process that has terminated at step~$i$ no longer appears in any sets $S_j$ for $j>i$.
A finite prefix  $\alpha=S_1,\dots,S_k$ of $\mathcal{S}$ is called an \emph{execution}. If $\alpha$ is a strict prefix then not all processes terminate in~$\alpha$. 
We say that execution $\alpha$ is an  execution \emph{by a set of processes} $P\subseteq [n]$ if $P=\bigcup_{i=1}^kS_j$. That is, $P$ is the set of processes taking steps in $\alpha$, or, said differently, $P$ is the \emph{participating set} of $\alpha$, and any process in $P$ is a \emph{participating process}.
An execution is \emph{complete} if every process in $[n]$ has terminated in $\mathcal{A}$ at the end of this execution.
We denote by $\operatorname{dec}(\alpha)$, for ``decide'', the set of output values produced by the processes that have terminated in $\mathcal{A}$ during execution $\alpha$, and, for $P\subseteq [n]$,  we denote by $\operatorname{dec}(\alpha,P)$ the set of outputs values produced by the processes of $P$ that terminated in~$\mathcal{A}$ during execution~$\alpha$.
A process $i\in [n]$ is \emph{unseen} in an execution $\alpha=S_1,\ldots S_k$ if $i$ appears in~$\alpha$ (i.e., there exists $j\in\{1,\dots,k\}$ such that $i\in S_j$), and all processes in $\alpha$ but process~$i$ have terminated before process~$i$ takes a step. That is, $\alpha=S_1,\ldots S_h,\{i\}^{h-k}$ where $1\leq h<k$, and, for every $j\in\{1,\dots,h\}$, $i \notin S_j$. 

The notion of order-preserving permutation plays a central role in the impossibility of weak symmetry-breaking. We denote by $\Sigma_n$ the set of all permutations on $n$ elements. 

\begin{definition}[\cite{AttiyaP16}]
A permutation $\pi\in\Sigma_n$ over $[n]$ is \emph{order preserving} on a set $P\subseteq [n]$ if, for every $(i,j) \in P \times P$, $i<j\Rightarrow \pi(i)<\pi(j)$. 
\end{definition}

Given an execution $\alpha$, $\pi(\alpha)$ corresponds to the execution where each occurrence of $i$ in every ``block'' $S_j$ in $\alpha$ is replaced by $\pi(i)$.

\begin{definition}[\cite{AttiyaP16}]
An algorithm is \emph{symmetric} if for every execution $\alpha$ on a subset $P\subseteq [n]$, and for every permutation $\pi\in \Sigma_n$ order preserving on $P$, it holds that, for every $i\in P$, process $i$ outputs $x$ in $\alpha$ if and only if process $\pi(i)$ outputs~$x$ on $\pi(\alpha)$.
\end{definition}

Note that we may consider executions where not all processes participate. In particular, if a single process participates in a symmetric algorithm, its output must be the same, whatever its identifier $i\in [n]$.

Recall that in its strong and weak versions, \emph{Symmetry Breaking} asks the processes to output a value~0 or~1. In \emph{Weak} Symmetry Breaking (WSB), the only constraint is that if \emph{all} $n$ processes output, then at least one outputs~0, and at least one output~1. \emph{Strong} Symmetry Breaking (SSB) introduces the additional constraint that at least one process outputs~1 in every execution in which at least one process participate. 

\begin{theorem}[\cite{AttiyaP16,BorowskyG93,CastanedaR10,GafniRH06}] ~
    \begin{itemize}
        \item For every $n\geq 2$, there are no wait-free algorithms solving SSB in an $n$-process asynchronous shared-memory system.
    
        \item For every $n\geq 2$ prime power (i.e. $n=p^k$ for some prime number~$p$, and some positive integer~$k$), there are no symmetric wait-free algorithms solving WSB in an $n$-process asynchronous shared-memory system.
    \end{itemize}
\end{theorem}

We shall extend this result to the case in which processes are given some inputs. We base our approach on the proof from~\cite{AttiyaP16}. So, let us recall some of the concepts defined there. For any algorithm~$\mathcal{A}$,  let $T(\mathcal{A})$ denote its \emph{trimmed} version, i.e., where every process stops if it has heard of all the $n$ processes in the system. More precisely, let $i$ be the first step at which every process has been activated at least once (i.e. $\cup_{j\le i}S_j=[n]$). There are then three possibilities for the output of each process in $T(\mathcal{A})$:
\begin{enumerate}
   \item Output 0 if its first activation was in $S_i$ (i.e., it sees all the processes at its first activation);
    \item Output 1 if it is activated in $S_j$ and $S_k$ for $j< i\le k$ (i.e., it did not output before step $i$);
    \item Output the same output as in $\mathcal{A}$ if it terminates before step~$i$.
\end{enumerate}
Given a scheduling $\mathcal{S}$, we can directly infer in which class a process falls in: 
\begin{itemize}
    \item All processes in $S_i\smallsetminus \cup_{j< i}S_j$ are class~(1); 
    \item The processes in $\cup_{j\geq i} S_j\smallsetminus (S_i\smallsetminus \cup_{j< i}S_j)$ are in class~(2);
    \item All the other processes are in class~(3). 
\end{itemize}
We denote by $SIM_{\alpha}$, for ``simulation'', the set of processes in classes~(2) or~(3) for execution~$\alpha$.

Note that there is a single execution of $T(\mathcal{A})$ where all processes fall in class~(1), namely  $\alpha_{all}=S_1=\{1,\dots,n\}$, i.e. the execution in which all processes are activated at the first step. As a consequence, 
\[
SIM_{\alpha}=\varnothing\iff\alpha=\alpha_{all}.
\]
For every $x\in\{0,1\}$, the set $\mathcal{C}^\mathcal{A}_x$ of \emph{$x$-univalued complete executions} of $\mathcal{A}$ is defined as the set of all the executions where all processes decide~$x$ in~$\mathcal{A}$, i.e., 
\[
\mathcal{C}^\mathcal{A}_x=\{\mbox{execution $\alpha\mid\alpha$ is complete for $\mathcal{A}$, and $\operatorname{dec}(\alpha)=\{x\}$}\}.
\]
As a consequence, one can show (cf.~\cite{AttiyaP16}) that 
\[
\alpha_{all}\in\mathcal{C}^{T(\mathcal{A})}_0, 
\;\mbox{and}\; 
\mathcal{C}^{T(\mathcal{A})}_1=\varnothing. 
\]
Let $\alpha=S_1\ldots S_k$ be an execution.
The \emph{sign} of $\alpha$ is defined as
\[
\sign{\alpha} \coloneqq \prod^k_{i=1}(-1)^{|S_i|+1}.
\]
Therefore $\sign{\alpha}=1$ if and only if $\alpha$ has an even number of ``blocks'' $S_i$ with odd size. Finally, the \emph{univalued signed count} of $\mathcal{A}$ on $n$ processes is 
\[
\sum_{\alpha\in\mathcal{C}^\mathcal{A}_0}\sign{\alpha} +
    (-1)^{n-1}\sum_{\alpha\in\mathcal{C}^\mathcal{A}_1}\sign{\alpha}.
\]

\begin{lemma}[\cite{AttiyaP16}]\label{lem:univalue}
For every algorithm $\mathcal{A}$, $\mathcal{A}$ and $T(\mathcal{A})$ have the same univalued signed count.
Moreover, if $\mathcal{A}$ satisfies that at least one process outputs~1 whenever at least one process outputs, then the univalued signed count of $T(\mathcal{A})$ is nonzero.
\end{lemma}

Lemma~\ref{lem:univalue} is enough to derive the impossibility to solve SSB, merely because the univalued sign of $\mathcal{A}$ is~0, due to the fact that if $\mathcal{A}$ solves WSB then, for every $x\in\{0,1\}$, $\mathcal{C}^\mathcal{A}_x=\varnothing$. 
For the case of WSB, the symmetry notion enables to define an equivalence relation between executions. Specifically, we have $$\alpha\sim\alpha'$$ if there exists  permutation $\pi\in \Sigma_n$ such that 
(1)~$\alpha'=\pi(\alpha)$, and 
(2)~$\pi$ is order preserving on $SIM_\alpha$ and its complement $\overline{SIM_\alpha}=[n]\smallsetminus SIM_\alpha$. 
For every execution~$\alpha$, there exist $\binom n m$ executions equivalent to~$\alpha$, where $m=|SIM_{\alpha}|$, and they all have the same sign. Now, on the one hand, it known that if $n$ is power of a prime~$p$, then $\binom n m\equiv0[p]$ for every $m\in\{1,\dots,n-1\}$. On the other hand, $SIM_\alpha\neq [n]$ because there is always a process in class~(1). Moreover, $\alpha_{all}$ is the unique execution resulting in an empty set~$SIM$. It follows that  the univalued signed count modulo $p$ of $T(\mathcal{A})$  is nonzero. This contradicts the fact that the univalued signed count of $\mathcal{A}$ is null, yielding the impossibility result.

\subsection{Weak Symmetry Breaking with Input}\label{ssec:wsb}

In this section, we extend the concepts introduced in the previous section to the case where processes have inputs. We thus consider input functions 
\[
\sigma:[n] \to [n]^*\times I,
\]
where $I$ is a set of possible inputs, $[n]^*$ is a (possibly empty) sequence of identifiers, and $\sigma(i)$ is the input assigned to process~$i$. In this way, input functions can provide each process with its neighboring processes in a graph, and with some input value, e.g., a coloring. An another example, input functions can provide a matching between the processes. 

We directly extend the notion of equivalence between execution to executions involving inputs. For every two executions~$\alpha$ and $\alpha'$, and for every two input functions~$\sigma:[n] \to [n]^*\times I$ and $\sigma':[n] \to [n]^*\times I$, we say that $$(\alpha,\sigma)\sim (\alpha',\sigma')$$ if there exists a permutation $\pi\in\Sigma_n$ such that 
(1)~$\alpha'=\pi(\alpha)$, 
(2)~$\pi$ is order preserving on $SIM_\alpha$ and its complement $\overline{SIM_\alpha}$, and 
(3)~$\sigma'=\pi\circ \sigma\circ\pi$, i.e., for every $i\in[n]$, if $\sigma(\pi(i))=((i_1,\dots,i_k),x)\in [n]^k\times I$ for some $k\geq 0$, then $\sigma'(i)= ((\pi(i_1),\dots,\pi(i_k)),x)$.

\subparagraph{Example: Cycle Input.}

Let $\sigma:[n]\to [n]^2$ be the input function defined as 
\[
\sigma(i)=(\sigma_l(i),\sigma_r(i))
\]
where $\sigma_l\in\Sigma_n$ and $\sigma_r\in\Sigma_n$ are two circular permutations satisfying $\sigma_l\circ \sigma_r= \mbox{Id}_{[n]}$. Let $\pi\in\Sigma_n$, and let $\sigma'=(\sigma'_l,\sigma'_r)$ be the input function defined by $\sigma'_l=\pi\circ \sigma_l\circ \pi$ and $\sigma'_r=\pi\circ \sigma_r\circ \pi$. 
For every execution $\alpha$, we have $(\alpha,\sigma)\sim (\pi(\alpha),\sigma')$ whenever $\pi$ is order preserving on $SIM_\alpha$. Note that there are $(n-1)!$ circular permutations over~$[n]$. 

\medbreak

We also extend the definition of univalued complete execution to pairs of executions-input functions by setting 
\[
\mathcal{C}^\mathcal{A}_x=\{(\alpha,\sigma) \mid \mbox{$\alpha$ is complete for $\mathcal{A}$, and $\operatorname{dec}(\alpha)=\{x\}$ with input~$\sigma$}\}.
\]
In particular, for any input function~$\sigma$, we have $(\alpha_{all},\sigma)\in\mathcal{C}^{T(\mathcal{A})}_0$. Also, we set
\[
\sign{\alpha,\sigma}=\sign{\alpha}, 
\]
and we reset the definition of univalued signed count of an algorithm~$\mathcal{A}$ with inputs accordingly. It is tedious but straightforward to check that Lemma~\ref{lem:univalue} still holds with these extensions to algorithms with inputs. 

\begin{definition}\label{def:goodinput}
Given a set~$I$ of input values, a set $\mathcal{I}=\{\sigma:[n]\to [n]^*\times I\}$ of input functions is \emph{\goodInput} if the following two conditions hold:
\begin{itemize}
\item If $n$ is prime, then $|\mathcal{I}|$ is not divisible by~$n$.
   \item For every permutation $\pi\in\Sigma_n$, and every input function $\sigma\in \mathcal{I}$, the input function $\sigma'=\pi\circ\sigma\circ\pi$ belongs to~$\mathcal{I}$.
\end{itemize}
\end{definition}

For instance, the set of Cycle inputs is \goodInput, because it has $(n-1)!$ elements. 
On the other hand, for every $k\in\{1,\dots,n-1\}$, if one considers the set of inputs where exactly~$k$ processes have input~1, and the $n-k$ others have input~0, then there are $\binom n k$ possible inputs, and thus this set of inputs is not non-prime-divisible. A typical example of an input function that is not order-invariant is  $\sigma(1)=$``leader'', and $\sigma(i)=$``defeated'' for every $i>1$. We can now state our main result in this section. 

\wsb*

\begin{proof}
    Let $n$ be a prime number, and let $\mathcal{I}$ be a \goodInput\/ set of inputs. Let $\alpha$ be an execution, let $\sigma\in \mathcal{I}$, and let us focus on the equivalent class of $(\alpha,\sigma)$. By definition, for any permutation $\pi\in\Sigma_n$ that is order-invariant on $SIM_\alpha$, any input function $\sigma'$ such as $(\alpha,\sigma)\sim (\pi(\alpha),\sigma')$ satisfies $\sigma'\in \mathcal{I}$. It follows that every pair $(\alpha,\sigma)$ in each equivalence class induces the same set $\operatorname{dec}(\alpha,\sigma)$ of outputs. As a consequence, the equivalence class of $(\alpha,\sigma)$ is of size $\binom n m$ where $m=|SIM_{\alpha,\sigma}|$. 
    
    Now, for every $m\in\{1,\dots,n-1\}$,  $\binom n m \equiv 0 \bmod n$. So, 
    every equivalence class of a pair $(\alpha,\sigma)\in\mathcal{C}^{T(\mathcal{(A)})}_0$ with $SIM_{\alpha,\sigma}\neq\emptyset$ contributes for~0 to the univalued signed count of $T(\mathcal{A})$.
    On the other hand, $SIM_{\alpha,\sigma}\neq [n]$, and we have that if $SIM_{\alpha,\sigma}=\varnothing$ then $\alpha=\alpha_{all}$. 
    Moreover, for every $\sigma\in \mathcal{I}$,  $(\alpha_{all},\sigma)\in\mathcal{C}^{T(\mathcal{(A)})}_0$.
    Therefore, the class of $(\alpha_{all},\sigma)$ contributes  to the univalued signed count of $T(\mathcal{A})$ for a number equal to the total number of input functions, which is different from~0 modulo~$n$. 
    Thanks to Lemma~\ref{lem:univalue} generalized with inputs, the univalued signed count of $\mathcal{A}$ is non-zero. This contradicts the fact that $\mathcal{A}$ solves weak symmetry-breaking, as such algorithms must have a zero univalued signed count.
\end{proof}

\subsection{Applications}

We now show how to apply Theorem~\ref{thm:wsb} for obtaining impossibility results in the \WFLOC model.  

\subsubsection{Impossibility of $4$-Coloring}
\label{sssec:4col}

\fourcol*

\begin{proof}
    Let us assume for the purpose of contradiction that there exists an algorithm $\mathcal{A}$ for 4-coloring the cycles in the \WFLOC model. This algorithm can be simulated in an asynchronous shared-memory system with a Cycle input, by having each process considering only the registers of its two neighbors in the cycles, as specified by the input. Given a 4-coloring of the input cycle, a process outputs~0 whenever its color is even, and it outputs~1 otherwise. By assumption, the size of the cycle is an odd prime natural number. In every odd cycle, every proper 4-coloring satisfies that there must be two nodes whose colors do not have the same parity.  It follows that the simulation of $\mathcal{A}$ solves WSB. As $n$ is prime, this contradicts Theorem~\ref{thm:wsb}.
\end{proof}

\subsubsection{Impossibility of Weak-$2$-Coloring}
\label{sssec:w2c}

\weaktwocol*

\begin{proof}
    The proof uses the same arguments as in the proof of Corollary~\ref{thm:4col}, by performing weak-2-coloring with colors~0 and~$1$, and mapping output color~$0$ (resp.,~1)  to the output~0 (resp.,~1) for WSB.
\end{proof}

\subsubsection{Impossibility of $(\Delta+2)$-Coloring}
\label{sssec:deltaplus3coloring}

\deltaplusthreecolor*

\begin{proof}
    Let $\Delta=2k$, and let $n>\Delta$ be prime. We construct a $\Delta$-regular graph $G$ on $n$ nodes, with chromatic number at least $k+2$. Let $u_0,\ldots,u_{n-1}$ be the $n$ nodes of~$G$. Each node $u_i$ is connected to the $k$ nodes $u_{i+1},\dots,u_{i+k}$ (the right neighbors), and to the $k$ nodes $u_{i-1},\dots,u_{i-k}$ (the left neighbors) --- the operations on the indices are performed modulo~$n$. The graph $G$ is indeed $\Delta$-regular. Every $k+1$  nodes with consecutive indices form a clique in~$G$, and thus $\chi(G)\geq k+1$. 
    However, $k+1$ colors do not suffice for a proper coloring of~$G$. Indeed, two nodes $u_i$ and $u_j$ such that $i\equiv j[k+1]$ must have the same color whenever using solely $k+1$ colors. Moreover, $u_{n-1}$ must have the same color as $u_{k}$, as $u_{n-1}$ is connected to the nodes $u_0,\ldots,u_{k-1}$. As $n$ is prime, and $n>k$, we have  $k \not\equiv n\pmod k$. It follows that the chromatic number of $G$ is at least $k+2$.

    Let us now consider a shared memory system with the following input  function $\sigma$. For every $i\in[n]$, 
    $
\sigma(i)=(\sigma_{-k}(i),\ldots,\sigma_{-1}(i),\sigma_1(i),\ldots\sigma_k(i))
$
where, for $j\in \{-k,\dots,-1,1,\dots,k\}$, $\sigma_j(i)$ is the identifier of the $j$th neighbor of processor~$i$ on the left if $j<0$, and on the right if $j>0$. Moreover, we add the restrictions that $\sigma_1$ is a circular permutation, and that, for every $j\in \{-k,\dots,-1,2,\dots,k\}$, $\sigma_j=(\sigma_1)^j$. This set of inputs is \goodInput, as it is fully defined by the circular permutation~$\sigma_1$.
Let us assume, for the purpose of contradiction, that there exists an algorithm $\mathcal{A}$ that $2k+2$-color~$G$. As $\chi(G)\geq k+2$, at least one process must output a color in $A=\{1,\dots,k+1\}$, and at least one node must output a color in $\{k+2,\dots,2k+2\}$. It follows that $\mathcal{A}$ enable to solve weak symmetry-breaking as follows: every node with a color in~$A$ outputs~0, and every node with a color in $B$ outputs~1. A contradiction. Therefore, $2k+2=\Delta+2$ colors are not sufficient to color~$G$ in \WFLOC.    
\end{proof}

\subsection{Impossibility of Weak MIS}
\label{ssec:wmis}

\weakmis*

\begin{proof}
    Let us assume for contradiction that there exists an algorithm $\mathcal{A}$ solving weak MIS in the cycles. Let us then consider two consecutive nodes $u$ and $v$ of the cycle, and executions in which they both terminate, but their other neighbors are never scheduled. Among such possible executions, their must be one for which both $u$ and $v$ output~0. Indeed, they cannot both output~1, and they always output~0 for one, and~1 for the other, then there would exist an algorithm solving WSB on 2~processes. 

    Let us now consider an $n$-node cycle with $n\geq 7$, and let us consider a sequence  $u_0$, $u_1$, $u_2$, $u_3$, $ u_4$, $u_5$, $ u_6$, $u_7$ of consecutive nodes in this cycle ($u_0=u_7$ if $n=7$). Let us consider an execution that does not activate nodes $u_0,u_3,u_4,u_7$,  but only $u_1,u_2$ and $u_5,u_6$ until they terminate. There is an execution in which $u_1$ and $u_2$ both output~0, and $u_5$ and $u_6$ both output~0. After $u_1,u_2,u_5$, and $u_6$ have terminated, let us assume that $u_3$ and $u_4$ are scheduled. One of these two nodes, say~$u_3$, must output~0. This implies that there are three consecutive nodes $u_1,u_2,u_3$ that all output~0, contradicting the fact that $\mathcal{A}$ solves weak MIS. 
\end{proof}

\subsection{Impossibility of $(\Delta+1)$-Coloring Trees}
\label{ssec:lbcol}

\lbcol*

\begin{proof}
    For the purpose of contradiction, let us assume the existence of an algorithm $\mathcal{A}$ solving $(\Delta+1)$-coloring in trees. The proof is based on a series of constructions in which the trees $T_1,\dots,T_k$ of a forest are connected to a node~$v$ so that to form a single tree~$T$. The node $v$ is actually connected to specific nodes $u_i\in V(T_i)$, $i=1,\dots,k$, whose outputs have special properties when  $\mathcal{A}$ is run in each $T_i$ separately. We shall argue that these properties still holds when each $u_i$ is connected to~$u$ (which modifies the neighborhood of~$u_i$, and thus may modify the action of each node~$u_i$ in~$\mathcal{A}$). The reason why this is valid argument is that node $u$ will always be scheduled after all nodes in each $T_i$ have terminated. Indeed, for the specific properties considered in the proof, if a better solution might be obtained for each $T_i$ whenever $u_i$ has an extra neighbor which is not scheduled, then a better solution could be obtained in $T_i$ alone, by having each node simulating the presence of an extra (non scheduled) node. 
    
Another important element is that we will duplicate the trees we construct to have several copies behaving the same way in different components of the graphs. However, we cannot directly do it when each process has a unique ID. We prove that for any ID assignment on some given tree, a coloring with some property will happen.

To bypass the ID issue, we argue that for any tree of size $n$, there exists an infinite collection of disjoint identifier sets of size $n$ and an ID assignment on this tree such that this assignment on the tree provides the same outcome coloring (with the given property).

For example, to $V=\{v_1,\ldots v_n\}$, for each $k\ge0$, we consider the ID assignment on set $\{kn+1,\ldots,(k+1)n\}$ that maps $v_i$ to the identifier $kn+i$. If we look, for each $k$, at each coloring produced, we know that there exists a coloring that appears infinitely many times. When we duplicate the tree for our recursive construction, we then use different set of identifiers for each copy.
    
    Keeping this preamble in mind, we can proceed with constructions, avoiding the tedious embedding of each tree in a larger tree whose nodes are not scheduled and choice of new adequate identifiers at each duplication of a previous construction.\\
    
    For every $\Delta\geq 1$, let $P_\Delta$ be the property: There exists a tree $T_\Delta$ of maximal degree $\Delta$, and a scheduling $\mathcal{S}_\Delta$ on $T_\Delta$ where $\Delta+1$ nodes $u_1,\ldots,u_{\Delta+1}$ of $T_\Delta$ output different colors in~$\mathcal{A}$. We prove by induction on~$\Delta$ that, for every $\Delta\geq 1$, $P_\Delta$ holds. 
    $P_1$ holds since a tree of maximum degree~1 is just an edge, and two adjacent nodes must output different colors.   
    Let us assume that $P_\Delta$ holds, and let us consider the forest $F_{\Delta+1}$ equal to $\Delta+1$ copies of $\mathcal{T}_\Delta$. Let $\mathcal{S}_{\Delta,j}$ be the scheduling  $\mathcal{S}_{\Delta}$ in the $j$th copy of $\mathcal{T}_\Delta$, and let us denote by $u_{1,j},\dots,u_{\Delta+1,j}$ the $\Delta+1$ nodes in the $j$th copy of $\mathcal{T}_\Delta$ outputting different colors in~$\mathcal{A}$, say $u_{i,j}$ outputs color~$i$. Let us add a node~$u$ connected to every node $u_{i,i}$, $i\in\{1,\dots,\Delta+1\}$, thus creating a tree $T_{\Delta+1}$. Let us then consider the scheduling 
    \[   \mathcal{S}_{\Delta+1}=\mathcal{S}_{\Delta,1},\ldots,\mathcal{S}_{\Delta,\Delta+1},\{u\}^k
    \] 
    where $k\geq 1$ is the number of activation needed for~$u$ to output a color. As $u$ has $\Delta+1$ colors in its neighborhood, $u$~must output another color, not in $\{1,\dots,\Delta+1\}$. Thus $P_{\Delta+1}$ holds. This concludes the proof of $P_\Delta$, for every $\Delta\geq 1$. 
      
    To complete the impossibility result, let us consider an edge $\{u,v\}$, and let us connect each extremity of this edge to $\Delta-1$ copies of $T_{\Delta-1}$. More specifically, $u$ (resp.,~$v$) is connected to the $\Delta-1$ copies of $T_{\Delta-1}$ denoted by $T^u_{\Delta-1,j}$, $j=1,\dots,\Delta-1$ (resp., $T^v_{\Delta-1,j}$, $j=1,\dots,\Delta-1$). We execute the appropriate scheduling for each of these trees, so that to get $\Delta$ nodes outputting different colors in each of them, say node $w^u_{i,j}$ colored $i$ in $T^u_{\Delta-1,j}$, and node $w^v_{i,j}$ colored $i$ in $T^u_{\Delta-1,j}$, for every $i\in\{1,\dots,\Delta\}$ and $j\in\{1,\dots,\Delta-1\}$. Node $u$ is connected to all nodes $w^u_{i,i}$ for $i\in\{1,\dots,\Delta-1\}$, and node $v$ is connected to all nodes $w^v_{i,i}$ for $i\in\{1,\dots,\Delta-1\}$. Again we consider the appropriate schedule for each of the trees, and $u$ and $v$ are activated only when all nodes in these trees have terminated. None of the two nodes $u$ and $v$ may use colors in $\{1,\dots,\Delta-1\}$, as they are taken by their neighbors. So, using $\mathcal{A}$, $u$ and $v$ are able to output two different colors in $\{\Delta,\Delta+1\}$. As a consequence, one could use $\mathcal{A}$ to solve WSB on 2~processes, a contradiction.
\end{proof}

\section{Open Questions}\label{sec:open}

We have shown that every $n$-node graph of maximum degree~$\Delta$ can be properly colored with $\frac12(\Delta+1)(\Delta+2)-1$ colors in $\WFLOC$, in $O(\log^\star n)+f(\Delta)$ rounds. The number of colors may seem large, but the \WFLOC model is considerably weaker than the (synchronous and failure-free) \LOCAL model. In particular, it is known that even the clique with $n=\Delta+1$ nodes cannot be colored with less than $2\Delta+1$ colors in \WFLOC (whenever $\Delta+1$ is power of a prime), and we have shown that there exists an infinite family of regular graphs with even degree~$\Delta$ that cannot be colored with less than $\Delta+3$ colors in~\WFLOC. One major question as far as solving graph problems in asynchronous crash-prone networks is thus the following. 

\medbreak 

\noindent\textbf{Open Problem:} Is there a $(2\Delta+1)$-coloring algorithm for graphs with maximum degree~$\Delta$ in the \WFLOC model, for every $\Delta\geq 2$? 

\medbreak 

Of course, if one puts aside the cliques, there might be a coloring algorithm for \WFLOC using a palette of less than $2\Delta+1$ colors. However, we have shown that, for $\Delta=2$, the bound $2\Delta+1=5$ is tight for infinitely many cycles. The only generic bound applying to infinitely many graphs of maximum degree~$\Delta$ is however only $\Delta+3$, so there might be room for improvement. Yet,  saving even just a single color in a palette of $\frac12(\Delta+1)(\Delta+2)$ colors was very delicate and difficult. So, progressing from a quadratic number of colors to a linear number of colors appears to be a challenge in \WFLOC. 

Finally, we question the efficiency of randomized algorithms in the \WFLOC model. 

\medbreak 

\noindent\textbf{Open Problem:} To which extent  randomized algorithms help in the \WFLOC model, in term of both complexity and computability? 

\medbreak 

\bibliographystyle{plainurl}
\bibliography{biblio}

\begin{thebibliography}{10}

\bibitem{AttiyaP16}
Hagit Attiya and Ami Paz.
\newblock Counting-based impossibility proofs for set agreement and renaming.
\newblock {\em J. Parallel Distributed Comput.}, 87:1--12, 2016.

\bibitem{AttiyaW04}
Hagit Attiya and Jennifer Welch.
\newblock {\em Distributed computing: fundamentals, simulations, and advanced
  topics}.
\newblock Wiley, 2004.

\bibitem{AwerbuchPPS92}
Baruch Awerbuch, Boaz Patt{-}Shamir, David Peleg, and Michael~E. Saks.
\newblock Adapting to asynchronous dynamic networks.
\newblock In {\em 24th ACM Symposium on Theory of Computing (STOC)}, pages
  557--570, 1992.

\bibitem{AwerbuchP90b}
Baruch Awerbuch and David Peleg.
\newblock Network synchronization with polylogarithmic overhead.
\newblock In {\em 31st IEEE Symposium on Foundations of Computer Science
  (FOCS)}, pages 514--522, 1990.

\bibitem{BarenboimEG18}
Leonid Barenboim, Michael Elkin, and Uri Goldenberg.
\newblock Locally-iterative distributed $(\delta+ 1)$-coloring below
  szegedy-vishwanathan barrier, and applications to self-stabilization and to
  restricted-bandwidth models.
\newblock In {\em 37\textsuperscript{th} {ACM} Symposium on Principles of
  Distributed Computing (PODC)}, pages 437--446, 2018.

\bibitem{BernardDPT09}
Samuel Bernard, St{\'{e}}phane Devismes, Maria~Gradinariu Potop{-}Butucaru, and
  S{\'{e}}bastien Tixeuil.
\newblock Optimal deterministic self-stabilizing vertex coloring in
  unidirectional anonymous networks.
\newblock In {\em 23rd {IEEE} International Symposium on Parallel and
  Distributed Processing (IPDPS)}, pages 1--8, 2009.

\bibitem{BlairM12}
Jean R.~S. Blair and Fredrik Manne.
\newblock An efficient self-stabilizing distance-2 coloring algorithm.
\newblock {\em Theoretical Computer Science}, 444:28--39, 2012.

\bibitem{BlinFB19}
L{\'{e}}lia Blin, Laurent Feuilloley, and Gabriel~Le Bouder.
\newblock Brief announcement: Memory lower bounds for self-stabilization.
\newblock In {\em 33rd International Symposium on Distributed Computing
  (DISC)}, volume 146 of {\em LIPIcs}, pages 37:1--37:3. Schloss Dagstuhl -
  Leibniz-Zentrum f{\"{u}}r Informatik, 2019.

\bibitem{BlinT13}
L{\'{e}}lia Blin and S{\'{e}}bastien Tixeuil.
\newblock Compact deterministic self-stabilizing leader election - the
  exponential advantage of being talkative.
\newblock In {\em 27th Int. Symp. on Distributed Computing (DISC)}, volume 8205
  of {\em LNCS}, pages 76--90. Springer, 2013.

\bibitem{BorowskyG93}
Elizabeth Borowsky and Eli Gafni.
\newblock Generalized {FLP} impossibility result for t-resilient asynchronous
  computations.
\newblock In {\em 25th ACM Symposium on Theory of Computing (STOC)}, pages
  91--100, 1993.

\bibitem{CastanedaFPRRT23}
Armando Casta{\~{n}}eda, Pierre Fraigniaud, Ami Paz, Sergio Rajsbaum, Matthieu
  Roy, and Corentin Travers.
\newblock Synchronous \emph{t}-resilient consensus in arbitrary graphs.
\newblock {\em Inf. Comput.}, 292:105035, 2023.

\bibitem{CastanedaR10}
Armando Casta{\~{n}}eda and Sergio Rajsbaum.
\newblock New combinatorial topology bounds for renaming: the lower bound.
\newblock {\em Distributed Comput.}, 22(5-6):287--301, 2010.

\bibitem{CastanedaDFRR19}
Armando Castañeda, Carole Delporte{-}Gallet, Hugues Fauconnier, Sergio
  Rajsbaum, and Michel Raynal.
\newblock Making local algorithms wait-free: the case of ring coloring.
\newblock {\em Theory of Computing Systems}, 63(2):344--365, 2019.

\bibitem{ColeV86}
Richard Cole and Uzi Vishkin.
\newblock Deterministic coin tossing and accelerating cascades: micro and macro
  techniques for designing parallel algorithms.
\newblock In {\em 18th ACM Symposium on Theory of Computing (STOC)}, pages
  206--219, 1986.

\bibitem{Delporte-Gallet19}
Carole Delporte{-}Gallet, Hugues Fauconnier, Pierre Fraigniaud, and
  Mika{\"{e}}l Rabie.
\newblock Distributed computing in the asynchronous {LOCAL} model.
\newblock In {\em 21st International Symposium on Stabilization, Safety, and
  Security of Distributed Systems (SSS)}, LNCS 11914, pages 105--110. Springer,
  2019.

\bibitem{coverfree}
P.~Erd{\"o}s, P.~Frankl, and Z.~F{\"u}redi.
\newblock Families of finite sets in which no set is covered by the union ofr
  others.
\newblock {\em Israel Journal of Mathematics}, 51(1):79--89, 1985.

\bibitem{FischerLP85}
Michael~J. Fischer, Nancy~A. Lynch, and Mike Paterson.
\newblock Impossibility of distributed consensus with one faulty process.
\newblock {\em J. {ACM}}, 32(2):374--382, 1985.

\bibitem{FraigniaudLR22}
Pierre Fraigniaud, Patrick Lambein{-}Monette, and Mika{\"{e}}l Rabie.
\newblock Fault tolerant coloring of the asynchronous cycle.
\newblock In {\em 36th Int. Symp. on Distributed Computing (DISC)}, volume 246
  of {\em LIPIcs}, pages 23:1--23:22. Leibniz-Zentrum f{\"{u}}r Informatik,
  2022.

\bibitem{FuchsK23}
Marc Fuchs and Fabian Kuhn.
\newblock List defective colorings: Distributed algorithms and applications.
\newblock In {\em 37th Int. Symp. on Distributed Computing (DISC)}, volume 281
  of {\em LIPIcs}, pages 22:1--22:23. Schloss Dagstuhl - Leibniz-Zentrum
  f{\"{u}}r Inf., 2023.

\bibitem{GafniRH06}
Eli Gafni, Sergio Rajsbaum, and Maurice Herlihy.
\newblock Subconsensus tasks: Renaming is weaker than set agreement.
\newblock In {\em 20th International Symposium on Distributed Computing
  (DISC)}, volume 4167 of {\em LNCS}, pages 329--338. Springer, 2006.

\bibitem{GhaffariK21}
Mohsen Ghaffari and Fabian Kuhn.
\newblock Deterministic distributed vertex coloring: Simpler, faster, and
  without network decomposition.
\newblock In {\em 62nd IEEE Symposium on Foundations of Computer Science
  (FOCS)}, pages 1009--1020, 2021.

\bibitem{GhaffariT23}
Mohsen Ghaffari and Anton Trygub.
\newblock A near-optimal deterministic distributed synchronizer.
\newblock In {\em 42th ACM Symposium on Principles of Distributed Computing
  (PODC)}, pages 180--189, 2023.

\bibitem{HalldorssonKNT22}
Magn{\'{u}}s~M. Halld{\'{o}}rsson, Fabian Kuhn, Alexandre Nolin, and Tigran
  Tonoyan.
\newblock Near-optimal distributed degree+1 coloring.
\newblock In {\em 54th ACM Symposium on Theory of Computing (STOC)}, pages
  450--463, 2022.

\bibitem{Herlihy2013}
Maurice Herlihy, Dmitry~N. Kozlov, and Sergio Rajsbaum.
\newblock {\em Distributed Computing Through Combinatorial Topology}.
\newblock Morgan Kaufmann, 2013.

\bibitem{HirvonenS2020}
Juho Hirvonen and Jukka Suomela.
\newblock {\em Distributed Algorithms}.
\newblock Creative Commons, 2020.

\bibitem{Linial92}
Nathan Linial.
\newblock Locality in distributed graph algorithms.
\newblock {\em {SIAM} J. Comput.}, 21(1):193--201, 1992.

\bibitem{Lynch96}
Nancy~A. Lynch.
\newblock {\em Distributed Algorithms}.
\newblock Morgan Kaufmann, 1996.

\bibitem{NaorS95}
Moni Naor and Larry~J. Stockmeyer.
\newblock What can be computed locally?
\newblock {\em {SIAM} J. Comput.}, 24(6):1259--1277, 1995.

\bibitem{Nowak0W19}
Thomas Nowak, Ulrich Schmid, and Kyrill Winkler.
\newblock Topological characterization of consensus under general message
  adversaries.
\newblock In {\em 38th {ACM} Symposium on Principles of Distributed Computing
  (PODC)}, pages 218--227, 2019.

\bibitem{Peleg2000}
David Peleg.
\newblock {\em Distributed Computing: A Locality-Sensitive Approach}.
\newblock SIAM, 2000.

\bibitem{Raynal18}
Michel Raynal.
\newblock {\em Fault-Tolerant Message-Passing Distributed Systems - An
  Algorithmic Approach}.
\newblock Springer, 2018.

\bibitem{SantoroW89}
Nicola Santoro and Peter Widmayer.
\newblock Time is not a healer.
\newblock In {\em 6th Annual Symposium on Theoretical Aspects of Computer
  Science (STACS)}, volume 349 of {\em LNCS}, pages 304--313. Springer, 1989.

\bibitem{WinklerPG0023}
Kyrill Winkler, Ami Paz, Hugo~Rincon Galeana, Stefan Schmid, and Ulrich Schmid.
\newblock The time complexity of consensus under oblivious message adversaries.
\newblock In {\em 14th Innovations in Theoretical Computer Science Conference
  (ITCS)}, volume 251 of {\em LIPIcs}, pages 100:1--100:28. Schloss Dagstuhl -
  Leibniz-Zentrum f{\"{u}}r Informatik, 2023.

\end{thebibliography}

\appendix

\centerline{\Large\bf A P P E N D I X}

\section{Example of an execution of an algorithm for 6-coloring cycles}
\label{app:example:alg:cyclesixcoloring}

An example of an execution of \Cref{alg:cyclesixcoloring} is provided in \Cref{tab:example}, in which the old and new states of each node after each step is displayed.

\begin{table}[htbp]
\resizebox{\textwidth}{!}{
\begin{tabular}{c | c | c|c| c|c| c|c| c|c |c }
& \multicolumn{2}{|c|}{3} & \multicolumn{2}{|c|}{5} & \multicolumn{2}{|c|}{4} & \multicolumn{2}{|c|}{1} & \multicolumn{2}{|c}{6} \\
                 &   Old & New  & Old & New  & Old & New  & Old & New  & Old & New \\\hline
  Initialization & $\bot$ & $(3,0,0)$ & $\bot$ & $(5,0,0)$ & $\bot$ & $(4,0,0)$ & $\bot$ & $(1,0,0)$ & $\bot$ & $(6,0,0)$ \\\hline
  $\{1,3,5\}$ & & & & & & & & & &\\
  after write & \boldmath $(3,0,0)$ & $(3,0,0)$ & \boldmath $(5,0,0)$ & $(5,0,0)$ & $\bot$ & $(4,0,0)$ & \boldmath $(1,0,0)$ & $(1,0,0)$ & $\bot$ & $(6,0,0)$ \\
  update & $(3,0,0)$ & \boldmath $(3,1,0)$ & $(5,0,0)$ & \boldmath $(5,0,1)$ & $\bot$ & $(4,0,0)$ & $(1,0,0)$ & \boldmath $T(0,0)$ & $\bot$ & $(6,0,0)$ \\\hline
  $\{4, 5\}$  & & & & & & & & & &\\
  after write & $(3,0,0)$ & $(3,1,0)$ & \boldmath $(5,0,1)$ & $(5,0,1)$ & \boldmath $(4,0,0)$ & $(4,0,0)$ & $(1,0,0)$ & $T(0,0)$ & $\bot$ & $(6,0,0)$ \\
  update & $(3,0,0)$ & $(3,1,0)$ & $(5,0,1)$ & \boldmath $T(0,1)$ & $(4,0,0)$ & \boldmath $(4,1,1)$ & $(1,0,0)$ & $T(0,0)$ & $\bot$ & $(6,0,0)$ \\\hline
  $\{3, 4\}$  & & & & & & & & & &\\
  after write & \boldmath $(3,1,0)$ & $(3,1,0)$ & $(5,0,1)$ & $T(0,1)$ & \boldmath $(4,1,1)$ & $(4,1,1)$ & $(1,0,0)$ & $T(0,0)$ & $\bot$ & $(6,0,0)$ \\
  update & $(3,1,0)$ & \boldmath $T(1,0)$ & $(5,0,1)$ & $T(0,1)$ & $(4,1,1)$ & \boldmath $T(1,1)$ & $(1,0,0)$ & $T(0,0)$ & $\bot$ & $(6,0,0)$ \\\hline
  $\{6\}$  & & & & & & & & & &\\
  after write & $(3,1,0)$ & $T(1,0)$ & $(5,0,1)$ & $T(0,1)$ & $(4,1,1)$ & $T(1,1)$ & $(1,0,0)$ & $T(0,0)$ & \boldmath $(6,0,0)$ & $(6,0,0)$ \\
  update & $(3,1,0)$ & $T(1,0)$ & $(5,0,1)$ & $T(0,1)$ & $(4,1,1)$ & $T(1,1)$ & $(1,0,0)$ & $T(0,0)$ & $(6,0,0)$ & \boldmath $(6,0,1)$ \\\hline
  $\{6\}$  & & & & & & & & & &\\
  after write & $(3,1,0)$ & $T(1,0)$ & $(5,0,1)$ & $T(0,1)$ & $(4,1,1)$ & $T(1,1)$ & $(1,0,0)$ & $T(0,0)$ & \boldmath $(6,0,1)$ & $(6,0,1)$ \\
  update & $(3,1,0)$ & $T(1,0)$ & $(5,0,1)$ & $T(0,1)$ & $(4,1,1)$ & $T(1,1)$ & $(1,0,0)$ & $T(0,0)$ & $(6,0,1)$ & \boldmath $T(0,1)$ \\\hline
\end{tabular}
}
\caption{An example of execution of \Cref{alg:cyclesixcoloring}, for the cycle $C_5$ with consecutive node identifiers $(3,5,4,1,6)$. The example corresponds to the scheduling $\mathcal{S}=\{1, 3, 5\}, \{4, 5\}, \{3, 4\}, \{6\}, \{6\},\dots$, and  $T$ stands for $\stateterm$. At each step, the states that are updated are highlighted in bold.}
\label{tab:example}
\end{table}

\section{A Counterexample for an Existing Algorithm for 5-Coloring Cycles}
\label{ssec:5cycle}
\label{app:algo-has-a-bug}

We merely exhibit an instance of 5-coloring $C_4$ for which the algorithm in~\cite{FraigniaudLR22} does not terminate\footnote{For the interested reader, we found this counterexample by implementing a simulator for the \WFLOC model. This simulator tests a given algorithm with random schedulings.}. The algorithm presented in~\cite{FraigniaudLR22} is shown in \Cref{alg:bug}. In \Cref{tab:loop}, we provide an example of execution where the algorithm loops forever.

\begin{algorithm}[htb]
\caption{The (erroneous) 5-coloring algorithm of \cite{FraigniaudLR22}} \label{alg:bug}
\begin{algorithmic}
\Procedure{FiveColoring}{$\id{v}$,$\inp{v}$}
  \State $x \gets \id{v}$; \; $a \gets 0$; \; $b \gets 0$
  \Comment{$s=(x,a,b)$ is the state of node $v$}
  \Forever
     \State $(s_1,s_2) \gets \imm(s)$ 
     \State $P^+ \gets \{ i \in\{1,2\} \mid s_i \neq \bot \land s_i.x > x \}$ \Comment{neighbors with larger id}
     \State $C^+ \gets \{x_i.a \mid i \in P^+\}\cup \{x_i.b \mid i \in P^+\}$ \Comment{$a$ and $b$ of neighbors with larger id}
     \State $C \gets \{ x_i.a \mid  i \in\{1,2\} \land s_i \neq \bot \}\cup \{x_i.b \mid  i \in\{1,2\} \land s_i \neq \bot \}$  \Comment{$a$ and $b$ of all neighbors}
     \If{ $a \notin C$} \Return $a$
     \Else \If{$b \notin C$}  \Return $b$
     \Else
        \State $a \gets \min \mathbb{N} \smallsetminus C^+$
        \State $b \gets \min \mathbb{N} \smallsetminus C$
     \EndIf
     \EndIf
  \Endforever
\EndProcedure
\end{algorithmic}
\end{algorithm}

\begin{table}[htbp]
\resizebox{\textwidth}{!}{
\begin{tabular}{c | c | c|c| c|c| c|c| c }
& \multicolumn{2}{|c|}{3} & \multicolumn{2}{|c|}{4} & \multicolumn{2}{|c|}{2} & \multicolumn{2}{|c}{1} \\
                 &   Old & New  & Old & New  & Old & New  & Old & New   \\\hline
  Initialization & $\bot$ & $(3,0,0)$ & $\bot$ & $(4,0,0)$ & $\bot$ & $(2,0,0)$ & $\bot$ & $(1,0,0)$ \\\hline
  $\{2,3,4\}$ & & & & & & & &\\
  after write &  \boldmath $(3,0,0)$ & $(3,0,0)$ & \boldmath  $(4,0,0)$ & $(4,0,0)$ & \boldmath $(2,0,0)$ & $(2,0,0)$ & $\bot$ & $(1,0,0)$  \\
  update &  $(3,0,0)$ & \boldmath $(3,1,1)$ &  $(4,0,0)$ & \boldmath 
 $(4,0,1)$ & $(2,0,0)$ & \boldmath  $(2,1,1)$ & $\bot$ & $(1,0,0)$  \\\hline

 $\{1,3,4\}$ & & & & & & & &\\
  after write &  \boldmath $(3,1,1)$ & $(3,1,1)$ & \boldmath  $(4,0,1)$ & $(4,0,1)$ & $(2,0,0)$ & $(2,1,1)$ &  \boldmath $(1,0,0)$ & $(1,0,0)$  \\
  update &  $(3,1,1)$ &  \boldmath  $(3,2,2)$ &  $(4,0,1)$ & \boldmath 
 $(4,0,2)$ & $(2,0,0)$ & $(2,1,1)$ &  $(1,0,0)$ & \boldmath  $(1,2,2)$ \\\hline

  $\{3,4\}$ & & & & & & & &\\
  after write &  \boldmath $(3,2,2)$ &   $(3,2,2)$ &  \boldmath $(4,0,2)$ & 
 $(4,0,2)$ & $(2,0,0)$ & $(2,1,1)$ &  $(1,0,0)$ &  $(1,2,2)$   \\
  update &    $(3,2,2)$ &   \boldmath $(3,1,1)$ &  $(4,0,2)$ & \boldmath 
 $(4,0,1)$ & $(2,0,0)$ & $(2,1,1)$ &  $(1,0,0)$ &  $(1,2,2)$  \\\hline

$\{3,4\}$ & & & & & & & &\\
  after write &  \boldmath $(3,1,1)$ &   $(3,1,1)$ &  \boldmath $(4,0,1)$ & 
 $(4,0,1)$ & $(2,0,0)$ & $(2,1,1)$ &  $(1,0,0)$ &  $(1,2,2)$   \\
  update &    $(3,1,1)$ &   \boldmath $(3,2,2)$ &  $(4,0,1)$ & \boldmath 
 $(4,0,2)$ & $(2,0,0)$ & $(2,1,1)$ &  $(1,0,0)$ &  $(1,2,2)$  \\\hline
\end{tabular}
}
\caption{An example of execution where Algorithm~\ref{alg:bug} loops, for a 4-cycle with nodes' identifiers $(3,4,2,1)$ in consecutive order. 
The example is for the scheduling $\{2, 3, 4\}, \{1, 3, 4\}, \{3, 4\}, \{3, 4\},\dots$. Observe that the state obtained after scheduling $\{1,3,4\}$ is the same state as the one obtained after the fourth step (when $\{3,4\}$ is scheduled for the second time). Therefore, there exists a scheduling that makes the algorithm looping forever.}
\label{tab:loop}
\end{table}

\end{document}